\documentclass[11pt,letterpaper]{article}
\usepackage[margin=1in]{geometry}

\usepackage{amsmath,amsthm,amssymb}
\usepackage{hyperref}
\usepackage[capitalise,noabbrev]{cleveref}
\usepackage{xspace}
\usepackage{url}
\usepackage[dvipsnames]{xcolor}
\usepackage{appendix}
\usepackage{babel}
\usepackage{algorithm}
\usepackage[noend]{algpseudocode}
\usepackage{enumitem}

\usepackage{multirow}
\usepackage{tablefootnote}
\usepackage{afterpage}
\usepackage{lipsum}
\usepackage{booktabs}
\usepackage{thm-restate}
\usepackage{graphicx}
\usepackage{subcaption}
\newsavebox{\imagebox}

\hypersetup{colorlinks=true,allcolors=blue}

\theoremstyle{plain}
\newtheorem{theorem}{Theorem}[section]
\newtheorem{lemma}[theorem]{Lemma}
\newtheorem{claim}[theorem]{Claim}

\newtheorem{fact}[theorem]{Fact}

\theoremstyle{definition}
\newtheorem{definition}[theorem]{Definition}
\theoremstyle{remark}

\newtheorem*{remark*}{Remark}

\DeclareMathOperator{\E}{\mathbb{E}}

\def\RR{{\mathbb{R}}}

\DeclareMathOperator{\poly}{poly}
\DeclareMathOperator{\polylog}{polylog}

\DeclareMathOperator{\MST}{MST}
\DeclareMathOperator{\OPT}{OPT}

\DeclareMathOperator{\dist}{dist}
\DeclareMathOperator{\Diam}{diam}
\newcommand{\col}{\ensuremath{\mathsf{color}}\xspace}
\newcommand{\cell}{\ensuremath{\mathsf{cell}}\xspace}

\newcommand{\indic}{\mathbb{I}}

\let\eps\varepsilon
\let\epsilon\varepsilon

\title{Streaming Algorithms for Geometric Steiner Forest}

\author{Artur Czumaj\thanks{Research partially supported by the Centre for Discrete Mathematics and its Applications (DIMAP), by EPSRC award EP/V01305X/1, by a Weizmann-UK Making Connections Grant, and by an IBM Award.
Email: \texttt{A.Czumaj@warwick.ac.uk} }\\
    University of Warwick
  \and Shaofeng H.-C. Jiang\thanks{Part of this work was done when the author was at the Weizmann Institute of Science and Aalto University. Partially supported by a startup fund from Peking University.
    Email: \texttt{shaofeng.jiang@pku.edu.cn}
  }\\
  Peking University
  \and Robert Krauthgamer\thanks{Work partially supported by ONR Award N00014-18-1-2364,
  the Israel Science Foundation grant \#1086/18,
  a Weizmann-UK Making Connections Grant,
  the Weizmann Data Science Research Center,
  and a Minerva Foundation grant.
    Email: \texttt{robert.krauthgamer@weizmann.ac.il}
  }\\
  Weizmann Institute of Science
  \and Pavel Vesel\'y\thanks{Part of this work was done when the author was at the University of Warwick.
Partially supported by European Research Council grant ERC-2014-CoG 647557.
PV was also partially supported by GA \v{C}R project 19-27871X and by Center for Foundations of Modern Computer Science (Charles Univ. project UNCE 24/SCI/008).
  	Email: \texttt{vesely@iuuk.mff.cuni.cz}}\\
  Charles University
}
\begin{document}

\maketitle
\begin{abstract}
We consider a generalization of the Steiner tree problem, the \emph{Steiner forest problem}, in the Euclidean plane: the input is a multiset $X \subseteq \RR^2$, partitioned into $k$ color classes $C_1, \ldots, C_k \subseteq X$. The goal is to find a minimum-cost Euclidean graph $G$ such that every color class $C_i$ is connected in $G$.
We study this Steiner forest problem in the streaming setting, where the stream consists of insertions and deletions of points to $X$.
Each input point $x\in X$ arrives with its color $\col(x) \in [k]$,
and as usual for dynamic geometric streams,
the input is restricted to the discrete grid $\{1,\ldots,\Delta\}^2$.

We design a single-pass streaming algorithm that uses $\poly(k \cdot \log\Delta)$ space and time,
and estimates the cost of an optimal Steiner forest solution
within ratio arbitrarily close to the famous Euclidean Steiner ratio $\alpha_2$
(currently $1.1547 \le \alpha_2 \le 1.214$).
This approximation guarantee matches the state-of-the-art bound for streaming Steiner tree, i.e., when $k=1$, and it is a major open question to improve the ratio to $1 + \epsilon$ even for this special case.
Our approach relies on a novel combination of streaming techniques, like sampling and linear sketching, with the classical Arora-style dynamic-programming framework for geometric optimization problems, which usually requires large memory and so far has not been applied in the streaming setting.

We complement our streaming algorithm for the Steiner forest problem
with simple arguments showing that any finite multiplicative approximation
requires $\Omega(k)$ bits of space.
\end{abstract}

\newpage

\section{Introduction}
\label{sec:intro}

We study combinatorial optimization problems in dynamic geometric streams, in the classical framework introduced by Indyk~\cite{Indyk04}.
In this setting, focusing on low dimension $d=2$,
the input point set is presented as a stream of
insertions and deletions
of points restricted
to the discrete grid $[\Delta]^2 := \{1, \ldots, \Delta\}^2$.
Geometric data is very common in applications
and has been a central object of algorithmic study,
in different computational paradigms (like data streams, property testing and distributed/parallel computing)
and various application domains (like sensor networks and scientific computing).
In particular, stream-processing of geometric data is motivated by the current demand
to analyze massive datasets using modest computational resources.
Research on geometric streaming algorithms has been very fruitful,
and in particular, streaming algorithms achieving
$(1+\epsilon)$-estimation (i.e., approximation of the optimal value)
have been obtained for fundamental geometric problems, such as
$k$-clustering~\cite{DBLP:conf/icml/BravermanFLSY17,FS05,hu2019nearly},\footnote{For clustering problems, streaming algorithms can often report the $k$ centers (but not the clusters themselves) in addition to estimating the optimal value.}
facility location~\cite{CLMS13},
maximum cut~\cite{FS05, LSS09, DBLP:conf/stoc/ChenJK23},
and minimum spanning tree (MST)~\cite{FIS08}.
Although such estimation algorithms approximate the optimal value
without reporting a feasible solution,
they can provide useful statistics very efficiently,
especially for massive datasets that change over time.

Despite the significant progress in geometric streaming,
many other (often similarly looking) fundamental geometric problems
are still largely open.
Spe\-cif\-i\-cal\-ly, for the Traveling Salesman Problem (TSP) and the Steiner tree problem,
which are cornerstones of combinatorial optimiza\-tion,
it is a major outstanding question (see, e.g., \cite{openquestion})
whether a streaming algorithm can match the $(1+\epsilon)$-approximation
known for the offline setting~\cite{DBLP:journals/jacm/Arora98,Mitchell99}.
In fact, the current streaming algorithms for TSP and Steiner tree
only achieve $O(1)$-approximation,
which follows easily by employing known streaming algorithms for MST on the same input.

While MST is closely related to TSP and to Steiner tree
-- the optimal MST value is within factor $2$ of the other two problems (for metric instances) --
it seems unlikely that techniques built around MST
could achieve $(1+\epsilon)$-approximation for either problem.
Indeed, even in the offline setting,
the only approach known to achieve $(1 + \epsilon)$-approximation
for TSP and/or Steiner tree relies on a framework devised independently
by Arora~\cite{DBLP:journals/jacm/Arora98} and by Mitchell~\cite{Mitchell99},
that combines geometric decomposition of the data points
and dynamic programming.
These two techniques have been used \emph{separately} in designing (one-pass) streaming algorithms in the past:
geometric decomposition, often using a quadtree,
in \cite{ABIW09,DBLP:conf/soda/AndoniIK08, CJLW22, CLMS13,FIS08,FS05, Indyk04, IT03,LS08}
and dynamic programming, mainly for string processing problems,
in \cite{BZ16, CGK16, CFHJLRSZ21, EJ15, GJKK07, SS13, SW07}.
However, we are not aware of any successful application of the Arora-Mitchell framework, which combines these two techniques,
for any geometric streaming optimization problem whatsoever.\footnote{We remark that \cite{ANOY14} does use a combination of dynamic programming and a quadtree, but not the Arora-Mitchell framework.
However, the context there is distributed computing, and it obtains streaming algorithms only for a restricted setting.}

We make an important step towards better understanding of these challenges
by developing new techniques that \emph{successfully adapt the Arora-Mitchell framework to streaming}.
To this end, we consider a generalization of Steiner tree, the classical \textbf{Steiner Forest Problem (SFP)}.
In this problem (also called \emph{Generalized Steiner tree}, see, e.g., \cite{DBLP:journals/jacm/Arora98}), the input is a multiset of $n$ \emph{terminal} points $X \subseteq [\Delta]^2$, partitioned into $k$ \emph{color classes} $X = C_1 \sqcup \cdots \sqcup C_k$, presented as a dynamic stream.
Thus, every point $x$ in the stream is described
by its coordinates and by its color $\col(x) \in [k]$.\footnote{The points are inserted and deleted in arbitrary order;
  there is no requirement that each color is presented in a batch,
  i.e., that its points are inserted/deleted consecutively in the stream.
}
The goal is to find a minimum-cost Euclidean graph $G$ such that every color class $C_i$ is connected in $G$.
We denote this cost, i.e., the optimal SFP value, by $\OPT$.
Observe that the Steiner tree problem is a special case of SFP in which all terminal points should be connected (i.e., $k=1$). Similarly to the Steiner tree problem, a solution to SFP may use points other than $X$; those points are called \emph{Steiner points}.
See \Cref{fig:sfp} for an illustration of the problem definition.

\begin{figure}[t]
  \centering
  \begin{subfigure}[c]{0.3\textwidth}
    \centering
    \includegraphics[width=\textwidth]{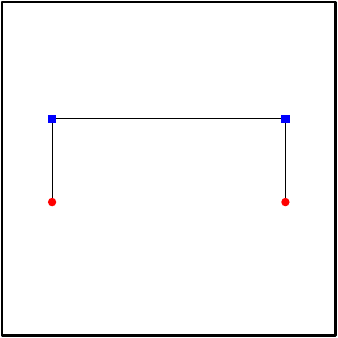}
  \end{subfigure}\qquad\qquad
  \begin{subfigure}[c]{0.3\textwidth}
    \centering
    \includegraphics[width=\textwidth]{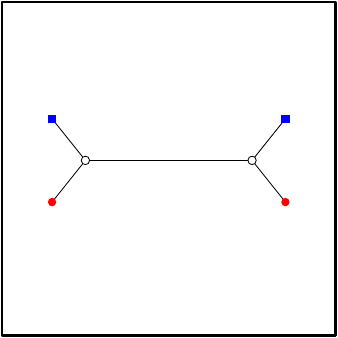}
  \end{subfigure}
\caption{An SFP instance with $k=2$ color classes (blue squares and red disks).
      The left figure shows a solution (straight lines) that forms a single connected component, which is cheaper than a separate component for each color.
  The right figure shows an even cheaper solution using Steiner points (black circles). 
  }
  \label{fig:sfp}
\end{figure}

\begin{remark*}
In the literature, the term SFP sometimes refers to the \emph{special case} where \emph{each color class contains only a pair of points}, i.e., each $C_i = \{s_i,t_i\}$ \cite{DBLP:journals/algorithmica/BateniH12,BHM11,DBLP:journals/talg/BorradaileKM15,CHJ18,GK15}.
In the standard setting of \emph{offline algorithms},
one can easily reduce one problem into another (see~\cite{Schaefer16}).
Thus, the special case of pairs is often simpler to present and does not restrict algorithmic generality for offline algorithms,
  while the definition used here is more natural in many applications,
  such as VLSI design (see~\cite[Section 2]{MW95} for more),
and allows for better parameterization over $k$, the number of colors.
\end{remark*}

\paragraph{Background.}

It is useful to recall here the \emph{Steiner ratio} $\alpha_2\ge 1$,
defined as the supremum over all finite point sets $X \subseteq \RR^2$,
of the ratio between the cost of an MST and that of an optimal Steiner tree.
The famous Steiner ratio Gilbert-Pollak Conjecture~\cite{gilbert1968steiner} speculates that $\alpha_2 = \frac{2}{\sqrt{3}} \approx 1.1547$, but the best upper bound to date is only that $\alpha_2 \le 1.214$ \cite{chung1985new}.
Hence, any $t$-approximation of MST in $\RR^2$ implies $(t \alpha_2)$-approximation for Steiner tree (which is a special case of SFP where $k = 1$).
In particular, employing the streaming algorithm of Frahling, Indyk, and Sohler~\cite{FIS08},
which $(1+\eps)$-approximates the MST cost using space $\poly(\eps^{-1}\log\Delta)$,
immediately yields a streaming algorithm that $(\alpha_2+\eps)$-approximates
the Steiner tree cost, with the same space bound.
Until our work, this simple $(\alpha_2 + \epsilon)$-approximation
was the only known streaming algorithm for Steiner tree,
and it was not known whether the same approximation can be obtained for SFP
(a more general problem).

\subsection{Our Results}
Our main result is a \emph{space and time} efficient,
\emph{single-pass} streaming algorithm that estimates the optimal \emph{cost}
for SFP within factor arbitrarily close to $\alpha_2$.
As usual, we analyze the streaming algorithm's space complexity,
and time complexity;
for the latter, we shall distinguish between two types of operations:
for inserting/deleting a point it is called \emph{update time},
and for reporting an estimate (of $\OPT$) it is called \emph{query time}.

\begin{restatable}{theorem}{thmmainprecise}
	\label{thm:main-precise}
	For every integers $k, \Delta\ge 1$, and every $0<\eps<1/2$,
	one can with high probability $(\alpha_2 + \eps)$-approximate
	the SFP cost of an input $X \subseteq [\Delta]^2$
	presented as a dynamic geometric stream,
	using space and update time $k^3\cdot \poly(\epsilon^{-1}\log k\log\Delta)$
	and query time $k^3\cdot \poly(\log k)\cdot (\epsilon^{-1}\log\Delta)^{O(\epsilon^{-2})}$. \end{restatable}

Our space bound is not too far from optimal
in terms of the dependence on $k$,
since any finite multiplicative approximation for SFP requires space $\Omega(k)$,
as we prove in Theorem~\ref{thm:lb} by a reduction from the communication complexity of indexing.
This lower bound holds already for insertion-only algorithms
and even in the one-dimensional case with all color classes of size at most $2$.
A polynomial dependence on $\log\Delta$ is common in this setting, 
and it is easy to prove an $\Omega(\log\Delta)$-bit lower bound,
even for the one-dimensional case and $k=1$, by reduction from augmented indexing. 
Furthermore, our approximation ratio matches the state of the art for Steiner tree (i.e., $k = 1$).

While our algorithm in Theorem~\ref{thm:main-precise} is an estimation algorithm,
i.e., it only approximates the optimal cost
and does not report an approximate solution (which has size $\Omega(n)$),
it can report how the colors are grouped in an approximate solution.
That is, our algorithm can report a partition of the colors $[k]=I_1\sqcup\cdots\sqcup I_r$,
so that the sum  (over $j=1,\ldots,r$) of the costs of
the minimum-cost Steiner trees for the sets $\bigcup_{i \in I_j} C_i$
is an $(\alpha_2+\eps)$-approximation of SFP.
It is worth noting that in estimating the optimal cost,
our algorithm does use Steiner points,
and thus the MST costs for sets $\bigcup_{i \in I_j} C_i$ of the aforementioned partition need not be within $(1+\eps)$-factor of the estimate of the algorithm.

\medskip

\begin{remark*}
Another streaming algorithm can be obtained 
by a simple brute-force enumeration combined with linear-sketching techniques,
as explained in Section~\ref{subsec:other_approaches}. 
The space complexity of this algorithm is close-to-optimal too,
but its query time is significantly worse and in fact \emph{exponential} in $k$.
Technically, this approach demonstrates the power of linear sketching,
however its core is an exhaustive search rather than an algorithmic insight,
and thus it is quite limited, offering no path for improvements or extensions.
Although the primary focus in streaming algorithms is on space complexity,
most applications require also low query time.
Indeed, exponential improvements in the query time of other streaming problems
have proved to be of key importance,
e.g., for moment estimation the query time was improved from
$\poly(\epsilon^{-1})$ to $\polylog (\epsilon^{-1})$~\cite{DBLP:conf/stoc/KaneNPW11},
and for heavy hitters from $\poly(n)$ to $\polylog (n)$ \cite{DBLP:journals/cacm/LarsenNNT19}.
\end{remark*}

\subsection{Related Work}
\label{subsec:related-work}

SFP has been studied extensively in operations research and algorithmic communities for several decades,
and is
part of a more general network-design problem,
where various subsets of vertices are required to maintain higher inter-connectivity,
see e.g.~\cite{AKR95,GW95,Jain01,MW95}.

In the offline setting, it is known that the Steiner tree problem is APX-hard in general graphs and in high-dimensional Euclidean spaces,
and the same thus holds for its generalization SFP. For SFP in general graphs,
a 2-approximation algorithm is known due to~\cite{AKR95}
(see also~\cite{GW95,Jain01}).
These 2-approximation algorithms rely on linear-programming relaxations,
and the only two combinatorial constant-factor approximations for SFP were devised more recently~\cite{GK15,gro_et_al:LIPIcs:2018:8313}.
For low-dimensional Euclidean space, which is the main focus of our paper,
Borradaile, Klein, and Mathieu~\cite{DBLP:journals/talg/BorradaileKM15}
and then Bateni and Hajiaghayi~\cite{DBLP:journals/algorithmica/BateniH12}
obtained a $(1+\epsilon)$-approximation by applying geometric decomposition and dynamic programming,
significantly extending the approach of Arora~\cite{DBLP:journals/jacm/Arora98}.
Further extensions of this approach
have led to a PTAS for SFP in metrics of bounded doubling dimension \cite{CHJ18},
and for planar graphs and graphs of bounded treewidth \cite{BHM11}.

There has been also extensive work on geometric optimization problems in the dynamic (turnstile) streaming setting.
Indyk~\cite{Indyk04} designed $O(d\log\Delta)$-approximation algorithms for several basic problems in $\RR^d$,
like MST and matching. 
Followup papers presented a number of streaming algorithms
achieving approximation ratio $1+\eps$ or $O(1)$ for the cost of Euclidean MST  \cite{FIS08}, various clustering problems \cite{FS05,HM04,hu2019nearly}, geometric facility location \cite{CLMS13,LS08,CJKVY22}, earth-mover distance \cite{ABIW09,Indyk04}, and various geometric primitives (see, e.g.,  \cite{AS15,AN12,Chan06,Chan16,FKZ04}).
Some papers have studied geometric problems with superlogarithmic but still sublinear space and in the multipass setting (see, e.g., \cite{ANOY14}).
A recent line of work~\cite{CJLW22,CJKVY22,WY22,DBLP:conf/stoc/ChenCJLW23,DBLP:conf/stoc/ChenJK23}
  focused on the high-dimensional geometric setting,
  where space complexity must be polynomial in the dimension $d$.
We are not aware of prior results that studied specifically
the (Euclidean) Steiner tree problem nor SFP in the streaming context,
but in $\RR^2$, as mentioned earlier,
$(1+\eps)$-approximation of the MST cost \cite{FIS08}
immediately gives $(\alpha_2+\eps)$-approximation for Steiner tree.

\section{Technical Overview: Adapting Arora's Framework to Streaming}
\label{sec:techoverview}

Our main technical contribution is to introduce the first streaming implementation
for the Arora-style (offline) framework \cite{DBLP:journals/jacm/Arora98}.\footnote{Arora~\cite{DBLP:journals/jacm/Arora98} and Mitchell~\cite{Mitchell99}
    have developed independently different versions of this framework. 
    We work with Arora's version, following previous work on SFP~\cite{DBLP:journals/talg/BorradaileKM15,DBLP:journals/algorithmica/BateniH12}.
}
More precisely, we rely on the offline algorithms
of Borradaile, Klein, and Mathieu~\cite{DBLP:journals/talg/BorradaileKM15}
and of Bateni and Hajiaghayi~\cite{DBLP:journals/algorithmica/BateniH12},
which extended Arora's framework to SFP,
and obtained a polynomial-time approximation scheme (PTAS).

At a high-level, Arora's framework achieves $(1 + \epsilon)$-approximation
by combining a hierarchical geometric decomposition with dynamic programming (DP).
The DP has $n$ basic subproblems (one for each input point),
and a streaming implementation of it would require, at best, space complexity $O(n)$.
We bypass this space requirement for SFP
by identifying a small set of new basic subproblems,
such that starting the DP from these subproblems
will achieve a $((1 + \epsilon)\alpha)$-approximation,
provided that we can $\alpha$-approximate each new basic subproblem.
This modified framework is built around two important steps:
identifying a small set of new basic subproblems, 
and solving (approximately) each of them.
Both have to be executed in streaming, and in fact in the same pass. 
Below, we first review briefly the offline algorithm of~\cite{DBLP:journals/talg/BorradaileKM15,DBLP:journals/algorithmica/BateniH12},
and then explain how we overcome the various challenges in a streaming implementation.

\paragraph{Review of Arora's framework for SFP \cite{DBLP:journals/talg/BorradaileKM15,DBLP:journals/algorithmica/BateniH12}.}
The key idea is to tweak an optimal solution so that
its cost remains nearly optimal but it now satisfies certain structural properties,
and then design a DP to find the best solution with this structure. 
The DP employs a standard tool, called a \emph{quadtree}
(see e.g.~\cite[Chapter~2]{Har-Peled-2011}),
which is a $4$-ary tree representing a hierarchical partitioning
(sometimes called dissection) of $(0, \Delta]^2$
(note that this partitions the underlying continuous space and not only the discrete grid),
where the nodes at level $i$ of the tree
represent a partition of $(0, \Delta]^2$ into (continuous) squares of sidelength $2^i$.
Specifically, the root node is at level $\log_2 \Delta$
and represents the entire $(0, \Delta]^2$;
each non-leaf node has four children representing four squares
that partition the parent's node square;
and the leaves at level $0$ represent squares of size $1\times 1$.
See \Cref{fig:quadtree} for illustration.

\begin{figure}[ht]
	\centering
	\begin{subfigure}[c]{0.2\textwidth}
		\centering
		\includegraphics[width=\textwidth]{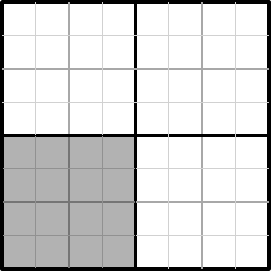}
	\end{subfigure}
	\quad
	\begin{subfigure}[c]{0.7\textwidth}
		\centering
		\includegraphics[width=\textwidth]{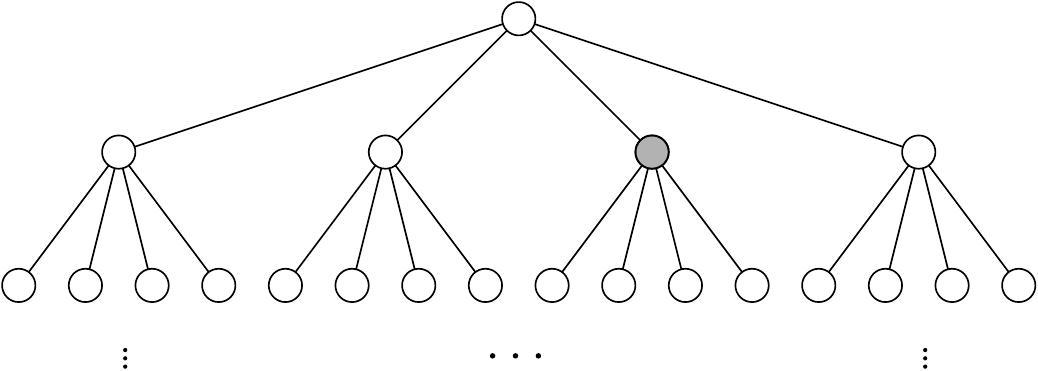}
	\end{subfigure}
	\caption{
            A hierarchical geometric decomposition with four levels (on the left) and its corresponding quadtree (on the right).
            In particular, the shaded square on the left corresponds to the shaded node in the quadtree.
          }
	\label{fig:quadtree}
\end{figure}

To compute a solution using DP over the quadtree,
we restrict the solution's structure inside every quadtree square $R$,
and in particular how the solution intersects the boundary of $R$. 
For the latter,
we designate $O(\eps^{-1}\cdot \log \Delta)$ equally-spaced points, called \emph{portals},
on all four boundary line segments of every quadtree square $R$.
Here, $\eps > 0$ is a parameter controlling the error introduced by imposing this structure.
The DP will consider only solutions that are \emph{portal-respecting},
meaning that they enter/exit every quadtree square only through its portals,
and the portals used in a square are called \emph{active portals}.
See Figure~\ref{fig:struct_prop} for illustration.

\begin{figure}[ht]
\centering
\includegraphics[width=0.3\textwidth]{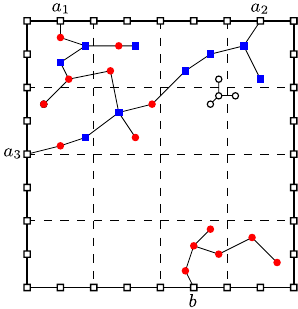}
\caption{A structured solution inside a square $R$
    that has portals (hollowed squares) on its boundary
    and is split into $4 \times 4$ cells (depicted by dashed lines).
    The solution inside $R$ is portal-respecting
    with active portals $a_1, a_2, a_3, b$. 
    It has three local components; 
two of them are connected to the boundary,
    and then by the cell property, points of same color (say red disks)
    in the same cell must be in the same local component. 
A third local component (whose points are depicted as hollowed circles)
    is entirely inside $R$,
    so the cell property does not apply to it. 
    For simplicity, local components are depicted without Steiner points (i.e., using MST).
  }
\label{fig:struct_prop}
\end{figure}

To restrict the solution's structure inside a quadtree square $R$,
we focus on the \emph{local components} of $R$,
i.e., the connected components of the intersection of the solution and $R$.
Notice that a \emph{global component},
i.e., a connected component of the entire solution,
may be split inside $R$ into multiple local components,
that are connected to each other outside $R$. 
There are two types of local components.
One is components completely contained in the interior of $R$,
i.e., not connected to the boundary,
which must be also global components of the solution.
Observe that a color class cannot intersect such a component only partially,
i.e., every color appearing in the component must be completely contained in it.
These components are usually handled by smaller subproblems in the DP.
Hence, we focus mostly on the other type,
of local components that are connected to the boundary of $R$.
We restrict their structure by partitioning the square $R$
into $O(\eps^{-1})\times O(\eps^{-1})$ subsquares called \emph{cells},
which are themselves quadtree squares at a lower level,
and requiring the solution to satisfy the so-called \emph{cell property}:
points that are in the same cell and are connected to the boundary of $R$
must be in the same local component of $R$.

Portal-respecting solutions that satisfy the cell property might be too costly.
However, if we apply a uniformly random shift on the quadtree, 
then with constant probability, there exists such a structured solution
whose cost is within $(1+\eps)$-factor of the optimal solution~\cite{DBLP:journals/talg/BorradaileKM15,DBLP:journals/algorithmica/BateniH12}; see Theorem~\ref{thm:dp_struct} for details.

We can design a DP that finds the cheapest solution with the above structure
(portal-respecting and satisfying the cell property). 
The details appear in Section~\ref{sec:dp_review},
but in a nutshell, each DP subproblem is specified by a quadtree square $R$,
a set of $O(\eps^{-1})$ active portals in $R$,
a mapping of each cell in $R$ to a set of active portals, 
and a partition of the active portals such that portals in the same part
must be connected outside of $R$.
The last datum ensures that certain local components inside $R$
will (eventually) be put together into the same global component.
The choice of parameters guarantees that 
the number of subproblems associated with each quadtree square $R$
is $\polylog(\Delta)$ (omitting dependence on $\epsilon$).

\paragraph{Streaming algorithm overview.}
Our overall plan is to maintain a sketch of the input while reading the stream,
and then use this sketch to evaluate the DP at the end of the stream. 
It turns out that the key is to identify a small set of basic DP subproblems
(that the DP will start from),
and then estimate the DP value of these subproblems, 
all executed in the streaming setting.
Our basic DP subproblems just correspond
to a small set of so-called \emph{simple} squares
(instead of all leaf squares, i.e., squares containing a single input point).
To approximate the value of a subproblem on a basic square,
we enumerate all partitions of the cells in this square,
and check if requiring that the cells in each part form a local component
would be compatible with the constraints imposed by the subproblem
and the colors of points inside the cells.
Next, we need to estimate the cost of each local component in this enumeration,
for which we use a $(1+\eps)$-approximate MST cost,
and this is the place that introduces $\alpha_2$ to our overall ratio.
Note that we use no Steiner point in these local components,
but portals of the squares might become Steiner points in the final solution.
We discuss below how to implement these steps in the streaming setting.

\paragraph{New definition of basic DP subproblems.}
The main difficulty in using Arora's approach in a streaming algorithm
is that it requires to evaluate too many basic subproblems,
because each input point lies in a separate (small enough) square,
so we essentially need to store the whole input.
To significantly reduce the number of basic subproblems,
it suffices to identify a small set of squares $\mathcal{R}$ that covers the input points,
evaluate the $\polylog(\Delta)$ subproblems for each square in $\mathcal{R}$,
and continue the original DP from these squares
(i.e., assuming all other squares are empty). 

Crucially, we enforce on these squares additional structural properties 
so that their associated subproblems can be evaluated in streaming.
The idea is to avoid local components
that are contained in the interior of $R$ and thus not connected to any portal. 
Specifically, a square $R$ is called \emph{simple}
if no color class is totally contained in $R$.
Observe that in a simple square $R$, all local components must be connected to the boundary;
furthermore, by the cell property,
points in the same cell of $R$ must be in the same local component.
We then identify a small set of simple squares that covers all points;
in fact, we just take all simple squares in the quadtree that are maximal
(with respect to containment).
Their number turns out to be only $O(k \log \Delta)$ (\Cref{lemma:leaves}),
and we further show (in \Cref{sec:streaming_simp_square})
how to find them in the streaming setting
using standard sparse recovery techniques.
See Figure~\ref{fig:simple_square} for an illustration.

\begin{figure}[t]
	\centering
	\includegraphics[width=0.3\textwidth]{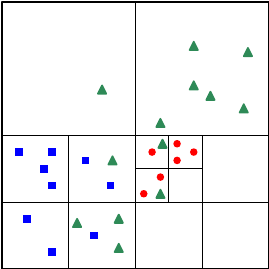}
	\caption{An input with three colors that has $13$ maximal simple squares (counting also the $4$ empty squares).}
	\label{fig:simple_square}
\end{figure}

These maximal simple squares are then our set of ``leaf'' squares
for the DP computation; 
that is, the basic DP subproblems will be all subproblems for these squares.
We discuss below how to estimate the values of these basic subproblems.
Finally, using these estimates, we run the DP 
in the same way as the offline algorithm~\cite{DBLP:journals/talg/BorradaileKM15,DBLP:journals/algorithmica/BateniH12}.

\paragraph{Enumerating local components.}
Given a simple square $R$ and a basic DP subproblem for $R$,
we find an approximately cheapest solution 
by enumerating all the possible ``arrangements'' of local components inside $R$.
By the cell property,
this boils down to enumerating all partitions of the cells of $R$.
However, not all partitions are compatible
with the DP subproblem, for one of the following reasons:
First, the assignment of cells to active portals in the subproblem
must agree with the partition,
for example, cells connected to the same portal must be in the same part of the partition. 
Second, two local components that share a color must be connected from outside of $R$,
hence the subproblem must specify two portals, one assigned to each component,
that are connected from outside of $R$.
While the former condition is simple to check given the subproblem and the partition,
for the latter the algorithm needs to know which color classes intersect every cell,
which is provably impossible in the streaming setting.
See \Cref{fig:comp_checking} for illustration. 

\begin{figure}[t]
	\centering
	\includegraphics[width=0.3\textwidth]{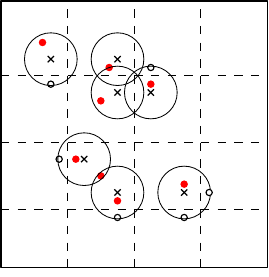}
	\caption{
            Approximately testing which color classes intersect every cell. 
            Input points (red disks) are covered by nearby grid points (crosses),
            and the artificial points added in close-enough cells (black circles), shown for simplicity only for cells containing no input points.
          }
	\label{fig:comp_checking}
\end{figure}

We therefore use the following approximate tester:
For every color class $C_i\subseteq X$, we compute a $\delta\cdot \Diam(C_i)$-covering
(defined in Section~\ref{sec:prelim} using grid points)
where $\delta>0$ is small enough and $\Diam(C_i)$ is the diameter of points in $C_i$.
We can approximately test which colors are present inside or near a cell
by only looking at the nearby points in these $k$ coverings.
More precisely, we declare that a cell intersects color $C_i$
if the covering for $C_i$ contains a point within distance $\delta\cdot \Diam(C_i)$ from the cell.
To make it more explicit, we add a new (artificial) point of color $C_i$ into the cell.
We can show that adding these points increases the SFP value only slightly (for small enough $\delta$),
and that these covering can be computed using standard streaming techniques,
namely, sparse recovery.

\paragraph{Estimating the MST cost for a local component.}
Our next goal is to estimate the MST cost of every local component
arising in the enumeration for a basic DP subproblem.
Note that by not using here Steiner points (except for active portals),
we lose a factor of $\alpha_2$;
however, improving upon this would require a streaming algorithm
with a better factor for Steiner tree, i.e., the case $k=1$.

For estimating the MST cost of a local component (given by a subset of cells and portals),
a natural idea is to employ the MST sketch designed in~\cite{FIS08} in a black-box manner.
This requires building an MST sketch for every possible local component
of every maximal simple square,
however the maximal simple squares are computed only at the end of the stream. 
We therefore modify the MST sketch from~\cite{FIS08} to handle \emph{subset} queries,
i.e., estimate the MST cost of points inside a subset of cells given as a query.
We first briefly overview this algorithm and then explain how we modify it.

\paragraph{Review of the MST sketch~\cite{FIS08}.}
The first observation (originally from~\cite{DBLP:journals/siamcomp/ChazelleRT05})
is that the MST cost can be written as a weighted sum
of the number of connected components in certain threshold graphs;
these graphs are obtained from the metric graph of the point set
by removing all edges whose length exceeds a threshold $\tau>0$.
Essentially, the idea is to count the number of MST edges of length above $\tau$.

To estimate the number of components in a threshold graph,
we round the points to a suitable grid
and sample a small number of rounded points uniformly, using $\ell_0$-samplers.
An $\ell_0$-sampler (cf.~\cite{FIS08,DBLP:journals/dpd/CormodeF14}) is a data structure that processes a dynamic stream (possibly containing duplicate items),
and with high probability
returns a uniform sample from the distinct items in the stream,
i.e., every item in the stream is chosen with the same probability $1/t$,
where $t$ is the $\ell_0$ norm of the resulting frequency vector
(see Lemma~\ref{lemma:neighbor_l0} for precise statement).
For each sampled (rounded) point $y$,
the algorithm in~\cite{FIS08} runs a stochastic-stopping BFS from $y$
to check if the entire component of $y$ is explored within a randomly chosen number of steps.
This BFS requires an extended $\ell_0$-sampler,
that returns for each sampled point also its neighboring points,
as presented in~\cite{FIS08} and stated in Lemma~\ref{lemma:neighbor_l0}.
The MST cost is estimated by a weighted sum of the number of BFS's that were completed,
summed over different thresholds $\tau>0$.

\paragraph{Adjusting the MST sketch to subset queries.}
We run an MST sketch for each color separately,
and when we need to estimate the MST cost inside a region
(a subset of cells and colors representing a local component,
together with some portals),
we first merge the MST sketches for the relevant colors, 
using the fact that the MST sketch is linear.\footnote{Linear sketches are merged by summing their memory contents,
  which requires they all use the same random coins. 
}
Then, out of all the (rounded) points sampled by the sketch, we only use those inside the region to estimate its MST cost.
That is, we run a stochastic-stopping BFS as in~\cite{FIS08} from each sampled point that falls inside the region, also restricting the BFS to stay inside the region.

However, the original analysis of the MST sketch in~\cite{FIS08}
has to be modified to deal with the restriction to the region. 
For instance, we might sample no points from the region
if that region contains only a few points,
in which case our MST-cost estimation will necessarily be $0$.
We observe that in this case the MST cost inside the region
is a small fraction of the overall cost,
which we can account for in a global way. 
Namely, we prove that each MST subset query is accurate up to an additive error
that is proportional to the MST cost of a \emph{global} component of a DP solution.
We can make this additive error small enough
so that the total additive error over all queries relevant to a global component
is at most an $\eps$-fraction of the component's cost.

Altogether, our adjustment of the MST sketch can answer MST subset queries
with small multiplicative and additive errors,
for subsets that satisfy some technical conditions.
This result is stated in \Cref{lemma:streaming_mst_add}
and may be of independent interest.

\subsection{Could Other Approaches Work?}
\label{subsec:other_approaches}

\paragraph{An exponential-time streaming algorithm based on linear sketching.}
An obvious challenge in solving SFP is to determine
the connected components of an optimal (or approximate) solution.
Each color class must be connected,
hence the crucial information is which \emph{colors} are connected together
(even though they do not have to be).
Suppose momentarily that the algorithm receives an advice with this information,
which can be represented as a partition of the color set
$[k]=P_1\sqcup\cdots\sqcup P_l$.
Then a straightforward approach for SFP is to solve the Steiner tree problem
separately on each part $P_j$ (i.e., the union of some color classes),
and report their total cost.
In our streaming model,
we can achieve $(\alpha_2+\eps)$-approximation the Steiner tree cost
by running the aforementioned MST algorithm~\cite{FIS08},
which requires space $\poly(\eps^{-1}\log\Delta)$,
and we would need only $l\leq k$ parallel executions of it (one for each $P_j$).
An algorithm can bypass having such an advice by enumeration,
i.e., by trying in parallel all the $k^k$ partitions of $[k]$
and reporting the minimum of all their outcomes.
This would still achieve $(\alpha_2+\eps)$-approximation,
because each possible partition gives rise to a feasible SFP solution
(in fact, this algorithm optimizes the sum-of-MST objective).
However, this naive enumeration increases the space and time complexities
by a factor of $O(k^k)$.
We can drastically improve the space complexity by the powerful fact
that the MST algorithm of~\cite{FIS08}
is based on a \emph{linear sketch}, i.e., its memory contents
is obtained by applying a randomized linear mapping to the input $X$
(viewed as a frequency vector in $\RR^{[\Delta^2]}$).
The huge advantage is that linear sketches of several point sets are mergeable,
by simply summing their memory contents.
In our context, one can compute a separate MST sketch for each color class $C_i$,
and then obtain a sketch for the union of some color classes,
say for some $P_j\subset [k]$, by merging their sketches.
Since these sketches are randomized,
one has to make sure they use the same random coins (same linear mapping),
and also amplify the success probability of the sketching algorithm
so as to withstand a union bound over all $2^k$ possible subsets of $[k]$.
This technique improves the space complexity and update time to be
$\poly(k\eps^{-1}\log\Delta)$,
however the query time is still \emph{exponential} in $k$.
See Theorem~\ref{thm:easy} for more details.

\paragraph{Tree embedding.}
Indyk~\cite{Indyk04} used the low-distortion tree embedding approach of Bartal~\cite{Bartal96}
to design dynamic streaming algorithms with approximation ratio $O(\log \Delta)$ for several geometric problems.
This technique can be easily applied to SFP as well, but the approximation ratio is $O(\log \Delta)$ which is far from what we are aiming at.

\paragraph{Other $O(1)$-approximate offline approaches.}
In the regime of $O(1)$-approximation, SFP has been extensively studied using various other techniques, not only dynamic programming. For example, in the offline setting there are several 2-approximation algorithms for SFP using the primal-dual approach and linear-programming (LP) relaxations \cite{AKR95,GW95,Jain01}, and there is also a combinatorial (greedy-type) constant-factor algorithm called \emph{gluttonous} \cite{GK15}. Both of these approaches work in the general metric setting. While there are no known methods to turn the LP approach into low-space streaming algorithms, the gluttonous algorithm of \cite{GK15} may seem amenable to streaming.
Indeed, it works similarly to Kruskal's MST algorithm as it builds components by considering the edges in order of increasing length.
Moreover, the MST cost estimation in~\cite{FIS08} is similar in flavor to Kruskal's algorithm.
However, a crucial difference is that the gluttonous algorithm stops growing a component once all terminals inside the component are satisfied,
i.e., no color class $C_i$ has points both inside and outside the component.
This creates a difficulty that the algorithm must know
for each component whether or not it is ``active'' (i.e., not satisfied),
and there are up to $n$ components, requiring overall $\Omega(n)$ bits of space. This information is crucial because ``inactive'' components do not have to be connected to anything else, but they may help to connect two still ``active'' components in a much cheaper way than by connecting them directly.
Furthermore,
we have a simple one-dimensional example showing that the approximation ratio of the gluttonous algorithm
cannot be better than $2$
(moreover, the approximation factor in \cite{GK15} is significantly larger than 2).
In comparison, our dynamic-programming approach gives a substantially better approximation ratio $\alpha_2+\eps$. Nevertheless, an interesting open question is whether the gluttonous algorithm admits a low-space streaming implementation.

 \section{Preliminaries}
\label{sec:prelim}

Let $\indic$ be the indicator function, i.e., $\indic(\mathcal{E}) = 1$ if and only if the condition $\mathcal{E}$ evaluates to true.  
For points $x, y \in \mathbb{R}^2$,
let $\dist(x, y) := \|x - y\|_2$.
For point sets $S, T \subset \mathbb{R}^2$,
let $\dist(S, T) := \min_{x \in S, y \in T}{\dist(x, y)}$,
and let $\Diam(S) := \max_{x, y \in S}{\dist(x, y)}$.
A \emph{$\rho$-packing} $S \subseteq \mathbb{R}^2$ is a point set
such that for all $x\neq y \in S$ one has $\dist(x, y) \geq \rho$.
A \emph{$\rho$-covering} of $X\subset \RR^2$ is a subset $S \subseteq \mathbb{R}^2$,
such that for all $x \in X$ there exists $y \in S$ with $\dist(x, y) \leq \rho$.

\begin{fact}[Packing Property, cf.~{\cite{pollard1990empirical}}] \label{fact:packing}
A $\rho$-packing $S \subseteq \mathbb{R}^d$ has size $|S| \leq \left( \frac{3\Diam(S)}{\rho} \right)^d$.
\end{fact}

	\paragraph{Metric graphs.}
	We call a weighted undirected graph $G=(X, E, w)$ a \emph{metric graph}
	if for every edge $\{u, v\} \in E$ one has $w(u, v) = \dist(u, v)$,
	and we let $w(G)$ be the sum of the weights of edges in $G$.
	A solution $F$ of SFP may be interpreted as a metric graph.
	For a set of points $S$ (e.g., $S$ can be a square), let $F_{|S}$ be the subgraph
	of $F$ formed by edges whose both endpoints belong to $S$.
	Note that we think of $F$ as a \emph{continuous} graph in which every point along an edge is itself a vertex,
	so $F_{|S}$ may be interpreted as a geometric intersection of $F$ and $S$.

\paragraph{Geometric streaming setting.}
	Recall that we assume the input points are from a discrete grid $[\Delta]^2$,
	and are presented as a dynamic point stream.
	This setting of geometric streaming was introduced in~\cite{Indyk04},
	and has become standard model for geometric streaming,
        see e.g.~\cite{FS05,FIS08,LS08,CLMS13,CJLW22,CJKVY22}.
	The discrete grid assumption is convenient
        for addressing and simplifying issues of precision and bit complexity.

	\paragraph{Randomly-shifted quadtrees~\cite{DBLP:journals/jacm/Arora98}.}
	Without loss of generality, let $\Delta$ be a power of 2, and let $L := 2\Delta$.
	A quadtree (hierarchical partitioning) is constructed on $(0, L]^2$.
	Each node $u$ in the quadtree corresponds to a (continuous) square $R_u$, and
        if it is not a leaf then it has four children, whose squares partition $R_u$.
        The squares in the quadtree have sidelengths that are powers of $2$,
	and a square $R_u$ is said to be at \emph{level} $i$ if its sidelength is $2^i$
	(this is also the level of its corresponding node $u$ in the quadtree);
        thus, the root is at level $\log_2 L$
        and tree nodes at level $0$ correspond to the smallest possible squares, of size $1\times 1$.
	To make sure the squares \emph{partition} the underlying space $(0, L]^2$,
	each square is a Cartesian product of two intervals of length $2^i$,
        where each interval is open on the left and closed on the right.
        For instance, the shaded square in \Cref{fig:quadtree} is $(0, 4] \times (0, 4]$.
    Finally, the squares of the entire quadtree are shifted by a uniformly random vector
    $v\in [-\Delta, 0]^2$,
    which is equivalent to shifting all input points by $-v\in[0, \Delta]^2$. 
	Throughout, we assume that the random shift $v$ has been sampled
	from the very beginning,
        and the input stream is independent of it (i.e., oblivious adversary). 
	
	We emphasize that a quadtree square $R$ refers to a continuous square,
	and contains points not from the input $X$.
	For $i = 0, \ldots \log_2{L}$, let the $2^i$-\emph{grid}
	$\mathcal{G}_i \subset \mathbb{R}^2$ be the set of centers of all level-$i$ squares in the quadtree,
        where the center of a square is its mean (which might have half-integral coordinates).

    \subsection{Review of Dynamic Programming for SFP~\cite{DBLP:journals/algorithmica/BateniH12,DBLP:journals/talg/BorradaileKM15}}
    \label{sec:dp_review}

    The (offline) PTAS for geometric SFP,
    due to~\cite{DBLP:journals/algorithmica/BateniH12,DBLP:journals/talg/BorradaileKM15}
    is based on the quadtree framework of Arora~\cite{DBLP:journals/jacm/Arora98}, with modifications tailored to SFP.
For a square $R$ in the (randomly-shifted) quadtree,
    let $\partial R$ be the boundary of $R$ (consisting of four line segments).
    Now for each $R$, 
    \begin{itemize} 
    \item designate $O(\epsilon^{-1}\log{L})$ equally-spaced points on $\partial R$
      as \emph{portals}; and
    \item designate the $\gamma \times \gamma$ subsquares of $R$ (at lower level)
      as the \emph{cells} of $R$, denoted $\cell(R)$,
      where $\gamma = \Theta(\epsilon^{-1})$ is a power of~$2$. 
    \end{itemize}
The following is the main structural theorem from~\cite{DBLP:journals/algorithmica/BateniH12}, see \cref{fig:struct_prop} for illustration.
    \begin{theorem}[\cite{DBLP:journals/algorithmica/BateniH12}]
        \label{thm:dp_struct}
        For an optimal solution $F$ of SFP, there is a solution $F'$ (defined with respect to the randomly-shifted quadtree),
        such that
        \begin{enumerate} 
			\item $w(F') \leq (1 + O(\epsilon)) \cdot w(F)$ with constant probability (over the randomness of the quadtree);
            \item For each quadtree square $R$, $F'_{\; |\partial R}$ has at most $O(\epsilon^{-1})$ components,
            and each component of $F'_{\; |\partial R}$ contains a \emph{portal} of $R$;
            \item For each quadtree square $R$ and each cell $P$ of $R$,
            if two points $x_1, x_2 \in X \cap P$ are connected 
            in $F'$ to $\partial R$, then they are connected in $F'_{\; |R}$;
            this is called the \emph{cell property}.
        \end{enumerate}
    \end{theorem}

    \paragraph{Defining DP subproblems.}
    It suffices to find a cheapest solution that satisfies the structure imposed by \cref{thm:dp_struct}.
    This is implemented using dynamic programming (DP), where a subproblem of the DP is identified as
    a tuple $(R, A, f, \Pi)$, specified as follows.\footnote{We follow the ideas from~\cite{DBLP:journals/algorithmica/BateniH12},
      although our notation and details differ slightly for sake of presentation.
    }
    \begin{itemize}  
        \item $R$ is a quadtree square.
        \item $A$ is a set of 
          $O(\epsilon^{-1})$ \emph{active portals}
          through which the local solution enters/exits $R$.
        \item $f : \cell(R) \to 2^A$ such that for all $S \in \cell(R)$, $f(S)$
        represents the active portals that the points inside $S$ connect to.
        We also require that for every two cells $S, S'\in \cell(R)$, either $f(S)=f(S')$ or $f(S)\cap f(S')=\emptyset$,
        and that every active portal in $A$ appears in $f(S)$ for some cell $S$.
        \item $\Pi$ is a partition of $A$, where active portals in each part of $\Pi$
		have to be connected outside of $R$ (in a larger subproblem).
	\end{itemize}
	The use of $R$ and $A$ is immediate, and $f$ is used to capture the
	connectivity between cells and portals (this suffices because we have the ``cell property'' in \cref{thm:dp_struct}).
	The additional requirements for $f$ do not impose real restrictions, but help to rule out $f$'s that are trivially infeasible.
        In particular, for the requirement that $f(S)$ and $f(S')$ are either disjoint or equal,
if these two sets intersect, 
then $S$ and $S'$ are in the same component and necessarily $f(S) = f(S')$.
Finally, $\Pi$ is used to ensure feasibility, since a global connected component may be broken into several local components in square $R$, and
	it is crucial to record whether or not these local components still need to be connected from outside of $R$.

	In fact, $f$ defines a partition of $\cell(R)$ and of $A$ into local components inside $R$ (taking into account only components connected to $\partial R$), and $\Pi$ should encode
	which local components need to be connected from the outside of $R$, implying that $\Pi$ should be a partition of local components
	(instead of $A$). Thus, $\Pi$ can be also thought of as a grouping of the parts in the partition of $A$ induced by $f$.

	An optimal solution for subproblem $(R, A, f, \Pi)$ is
	defined as a minimum weight metric graph in $R$
	that satisfies the constraints $A, f, \Pi$.
Standard combinatorial bounds show that the number of subproblems associated with each square is bounded by $(\epsilon^{-1}\log \Delta)^{O(\epsilon^{-2})}$ (see \cite{DBLP:journals/algorithmica/BateniH12}).

\paragraph{Evaluating DP subproblems.}
	Given a subproblem $(R, A, f, \Pi)$, we briefly describe how to evaluate its value in the offline setting.
The base case is when $R$ is a $1\times 1$ square, and the evaluation of the subproblem is trivial,
	as this subproblem has only a constant number of data points, hence one can brute-force enumerate all its possible solutions in small space and time.
	For the general case, let squares $R_1, \ldots, R_4$ be the four children of $R$.
	We enumerate all the possible subproblems of $R_i$'s, namely, tuples $(R_i, A_i, f_i, \Pi_i)$ for $i = 1, \ldots, 4$,
	and enumerate all possible edge sets of a metric graph $G$ with vertex set $\cup_i A_i$
        that connects active portals of child squares from outside (but still in $R$).
	This graph $G$, together with the information from tuples $(A_i, f_i, \Pi_i)$,
	defines how cells from $\cup_i \cell(R_i)$ are connected to the active portals in $\cup_i A_i$,
	and this can be used to deduce how the larger cells in $\cell(R)$ are connected to $A$.
	We check if this connectivity is compatible with $(A, f, \Pi)$, 
	particularly if it is consistent with the partition induced by $f$,
	and if components that are required to be connected outside the $R_i$'s  
	(according to the $\Pi_i$'s) are either already connected by $G$,
	or are still required to be connected outside $R$ (according to $\Pi$).
	The DP value for $(R, A, f, \Pi)$ is then taken to be the minimum,
	of the sum of the values of the subproblems and the weight of $G$,
	provided they are compatible with $(R, A, f, \Pi)$.
	This entire enumeration and compatibility checking process can be done
	in $(\epsilon^{-1} \log \Delta)^{\poly(\epsilon^{-1})}$ time and space.

To obtain the final solution, one should take the subproblem defined on the root square (which bounds the entire input $X$), 
	and require no active portals, hence $f$ is trivial and the constraint $\Pi$ is empty.

 \section{Streaming Dynamic Programming: $k^3$ Query Time and Space}\label{sec:dp_made_streaming}
	
In this section, we prove our main result, Theorem~\ref{thm:main-precise}, restated below with more precise bounds.
Formally, the \emph{update time} is the time bound for processing the insertion/deletion of one point,
and the \emph{query time} is that for reporting an estimate of $\OPT$.
	
\thmmainprecise*

\paragraph{Overview.}
Our approach for the streaming algorithm relies on a novel modification of the known PTAS for SFP in the offline setting \cite{DBLP:journals/algorithmica/BateniH12, DBLP:journals/talg/BorradaileKM15}, which is based on dynamic programming (DP). One important reason why the DP requires $\Omega(n)$ space is that $\Omega(n)$ leaves in the quadtree have to be considered as basic subproblems which correspond to singletons. To make the DP use only $\poly(k\cdot \log \Delta)$ space, we will use only $O(k\cdot \log \Delta)$ leaf nodes. Then, since each internal node in the quadtree has degree~$4$, the total number of squares to consider is $O(k\cdot \log \Delta)$. Furthermore, we design an algorithm that runs in time and space $\poly(k\cdot \log \Delta)$ and finds an $(\alpha_2 + \epsilon)$-approximate estimation for each new leaf and each DP subproblem associated with it. Finally, we apply the DP using such leaves as basic subproblems to obtain the estimation.
We start in \cref{sec:offline} with a description of this approach in the offline setting, and then present its streaming implementation in \cref{sec:streaming_simp_square,sec:streaming_comp_check,sec:streaming_dp}.
We then complete the proof of Theorem~\ref{thm:main-precise} in Section~\ref{sec:proof_kd}.

    \subsection{Offline Algorithm}
    \label{sec:offline}
    \paragraph{New definition of basic subproblems.}
    Each of our new leaves in the DP will be a \emph{simple square}
    defined below.
    The idea behind the definition is also simple: If no color class is contained in $R$, then all data points inside $R$ must be connected to $\partial R$,
    so we can make better use of the cell property in Theorem~\ref{thm:dp_struct}.

    \begin{definition}[Simple squares]
        \label{def:simple_sq}
        We call a square $R$ \emph{simple}
        if for every $1 \leq i \leq k$, $C_i \cap R \neq C_i$.
        In other words, there is no color class totally contained in $R$.
    \end{definition}

    We note that the number of all possible simple squares can still be large (in particular, any empty square is simple as well as any square containing a single point of color $C_i$ with $|C_i|\ge 2$).
    We use \cref{lemma:leaves} below to show the existence
	of a small subset of simple squares that covers the whole instance and can be found efficiently;
	these are in fact the maximal simple squares in the quadtree.
	Our new leaves are naturally defined using such subset of squares.

    \begin{lemma}
        \label{lemma:leaves}
        There is a subset $\mathcal{R}$ of disjoint simple squares,
        such that the union of the squares in $\mathcal{R}$
        covers $X$, and $|\mathcal{R}| = O(k\cdot \log \Delta)$.
    \end{lemma}
    \begin{proof}
        Consider the recursive procedure specified in Algorithm~\ref{alg:simpleSquares} that takes as input a square $R$ and returns a set of disjoint simple squares $\mathcal{R}$ that covers $R$; see \cref{fig:simple_square} for an illustration of the outcome of the procedure.
        For our proof, we apply the procedure with $R$ being the root square covering the whole instance.
        Suppose the procedure returns $\mathcal{R}$.
        \begin{algorithm}[ht]
            \caption{Algorithm for finding simple squares}
            \label{alg:simpleSquares}
            \begin{algorithmic}[1]
                \Procedure{Simp-Square}{$R$}
                    \If{$R$ is simple}
                            {return $\{R\}$}
\Else
                        \State let $\{R_i\}_i$ be the child squares of $R$ in the quadtree
                        \State \textbf{return} $\bigcup_{i}\textsc{Simp-Square}(R_i)$
                    \EndIf
                \EndProcedure
            \end{algorithmic}
        \end{algorithm}
    	
        We call a square $R$ intermediate square if it is a square visited
        in the execution of the algorithm and it is not simple (i.e., $R$ contains a color class).
        We observe that $|\mathcal{R}|$ is $O(1)$ times the number of intermediate squares.
        On the other hand, each color $C_i$ can be totally contained in
        at most $O(\log \Delta)$ intermediate squares.
        Therefore, $|\mathcal{R}| = O(k\cdot \log \Delta)$.
    \end{proof}

    \paragraph{Approximation algorithm for subproblems on simple squares.}
    Fix some simple square $R$. We now describe how each DP subproblem $(R, A, f, \Pi)$ associated with $R$
    can be solved directly using an $\alpha_2$-approximate algorithm that is amenable to the streaming setting.

        Since $R$ is a simple square, every point in $R$ has to be connected
        to the outside of $R$, as otherwise the color connectivity constraint is violated.
        Hence, by the cell property of \cref{thm:dp_struct}, for every cell $R' \in \cell(R)$,
        all points in $R'$ are connected in $R$.
Therefore, we enumerate all possible partitions of $\cell(R)$ that is consistent with the $f$ constraint,
i.e., that for each part of the partition, $f$ maps all
        	cells in the part to the same set of portals.
For each partition, we further check whether it satisfies the constraint defined by $\Pi$.
        To do so, for each cell $R' \in \cell(R)$, we scan through all
        colors, and record the set of colors $\mathcal{C}_{R'} \subseteq \mathcal{C}$ that intersects $R'$.
        The $\mathcal{C}_{R'}$'s combined with the $f$ constraint as well as the enumerated connectivity between cells suffice for checking the $\Pi$ constraint.
		In more detail, if there are two cells $R'$ and $R''$ in different parts (of the partition of cells) such that 
       	both $R'$ and $R''$ contain a point of one color, then there must be two portals $a'$ and $a''$
       	with $a'\in f(R')$ and $a''\in f(R'')$ that are in the same part
       	of the partition $\Pi$.

        Observe that every feasible solution of the subproblem corresponds to
        the above-mentioned partition of cells.
        Therefore, to evaluate the cost of the subproblem,
        we evaluate the sum of the MST costs of the parts in each partition
        and return the minimum one.
The time complexity of evaluating each subproblem is bounded since $|A| = O(\epsilon^{-1})$ and $|\cell(R)| = O(\epsilon^{-2})$.
        The approximation ratio is $\alpha_2$ because we use MST instead of Steiner tree for evaluating the cost.
        Using MST will enable us to implement this algorithm in the streaming setting.

	\subsection{Constructing Simple Squares in Streaming}
	\label{sec:streaming_simp_square}
	The first step is to construct a set of simple squares,
	as in \cref{lemma:leaves}, and an offline construction is outlined in \cref{alg:simpleSquares}.
	For the streaming construction of simple squares,
	we observe that the key component of \cref{alg:simpleSquares} is
	a subroutine that tests whether a given square is simple or not.
	To implement this subroutine, we compute the smallest bounding quadtree square for each color, and we test whether a given square contains any bounding square as a subsquare (observing that the given square is simple if and only if it does not contain a bounding square of any color). Specifically, we make use of a streaming sparse recovery structure (see \Cref{lemma:streaming_recover})
		to find the minimal bounding square for each color.
	\begin{lemma}
		\label{lemma:streaming_simp_square}
		\cref{alg:simpleSquares} can be implemented in the streaming setting,
		using space $O(k \polylog \Delta)$ and in time $O(k \polylog \Delta)$ per stream update, with success probability at least $1 - \poly(\Delta^{-1})$.
	\end{lemma}
	
	\begin{proof}
		We show that for a point set $S \subseteq [\Delta]^2$ and a quadtree decomposition,
		the smallest quadtree square that contains $S$ as a subset can be computed
		in the dynamic streaming setting.
		\begin{algorithm}[t]
			\caption{Finding the minimal enclosing quadtree square}
			\label{alg:min_square}
			\begin{algorithmic}[1]
				\Procedure{Min-Square}{$S \subseteq [\Delta]^2$}
					\Comment{$S$ presented as a dynamic point stream}
					\State for each $1 \leq i \leq \log_2{L}$,
					maintain sketch $\mathcal{K}_i$ of \cref{lemma:streaming_recover} for the $2^i$-grid $\mathcal{G}_i$, with parameter $T = 1$
					\For{insertion/deletion of point $x$ in the stream}
						\For{$i \gets 0, \ldots, \log_2{L}$}
							\State let $y \in \mathcal{G}_i$ be the grid point
							for the level-$i$ square that $x$ belongs to

							\Comment{recall that grid points in $\mathcal{G}_i$ are centers of level-$i$ squares}
							\State increase/decrease the frequency of $y$ by
							$1$ in sketch $\mathcal{K}_i$
						\EndFor
					\EndFor
					\Comment{the stream terminates, and the computing phase starts}
					\State find the least $i$ where
					$\mathcal{K}_i$ returns exactly one element,
					and \textbf{return} the corresponding square
				\EndProcedure
			\end{algorithmic}
		\end{algorithm}
		Namely, we use \cref{alg:min_square}. Next, we argue that this algorithm indeed finds the minimal enclosing quadtree square for a point set $S$,
w.h.p.
		Suppose $R$ is the minimal quadtree square that contains $S$, and suppose $R$ is of level $i$.
		Then all points in $S$ correspond to the same grid point $y$ in $\mathcal{G}_i$, so $\mathcal{K}_i$ returns exactly square $R$ w.h.p.
		On the other hand, for any level $i' < i$, sketch $\mathcal{K}_{i'}$
		either reports that the number of non-empty grid squares is larger than $2$,
		or returns 2 squares w.h.p. Hence
\cref{alg:min_square}
		finds
the minimal enclosing square of $S$ w.h.p.

		We apply \cref{alg:min_square} for each color $C_i \in \mathcal{C}$.
		After the stream ends, we get the bounding squares $\{R_i \}_i$ for colors $\{C_i\}_i$,
and we execute \cref{alg:simpleSquares} using $\{R_i\}$.
		In particular, to test whether a square $R$ is simple or not,
		it suffices to scan through all the bounding squares $\{R_i\}_i$,
		and $R$ is simple if and only if $R$ does not contain any $R_i$ as a subsquare
		(note that if $R = R_i$ for some $i$, then $R$ is not simple).
By the guarantee of \cref{lemma:streaming_recover},
		the space is bounded by $O(k \polylog \Delta)$,
		and the overall success probability is at least $1 - \poly(\Delta^{-1})$.
\end{proof}

	\subsection{Approximate Compatibility Checking in Streaming}
	\label{sec:streaming_comp_check}
	Suppose we applied \cref{lemma:streaming_simp_square} to obtain a set of
	simple squares $\mathcal{R}$.
	We proceed to evaluate the cost of subproblems associated with each simple square.
	Fix a simple square $R \in \mathcal{R}$.
	We next describe how to evaluate the cost for every subproblem associated with $R$, in streaming.

	Suppose we are to evaluate the cost of a subproblem $(R, A, f, \Pi)$.
	Since $R$ is known, we have access to $\cell(R)$, and hence, we can enumerate
	the connectivity between the cells, which is a partition of $\cell(R)$, on-the-fly without maintaining other information about the input.
	(Note that, by \Cref{thm:dp_struct}, there are at most $O(\eps^{-1})$ components inside $R$, 
	so we only consider partitions of $\cell(R)$ with at most $O(\eps^{-1})$ parts.)
	We can also check the compatibility of the partition of cells with the $f$ constraint,
	since the constraint only concerns the information about $A$ and the partition.
	Then, to check the compatibility of the partition of cells with $\Pi$, 
	we need to compute the set of colors $\mathcal{C}_{R'} \subseteq \mathcal{C}$ that a cell $R'$ intersects.

	However, computing this set $\mathcal{C}_{R'}$ is difficult in the streaming setting, even
	if there is only one color $C$. Indeed, testing whether color $C$ has an intersection with cell $R'$
	can be immediately reduced to the INDEX problem (see e.g.\ \cite{KushilevitzNisan97} or \cref{sec:lb} for the definition),
	which implies an $\Omega(n)$ space bound, where $n$ is the number of points of color $C$.
	Therefore, we need to modify the offline algorithm, and only test the intersection approximately.
	
	To implement the approximate testing, we use the following process to add some new points to the input and obtain an augmented instance $(X', \mathcal{C}')$.
	Specifically, for every color $C \in \mathcal{C}$,
	we impose a $\delta \cdot \Diam(C)$-covering $N_C$ for $C$
	(see \Cref{sec:prelim} for the definition of covering), 
	where $\delta := O(\epsilon^3 (k\log \Delta)^{-1})$.
	Hence, using $D_C := \Diam(N_C)$, we have that $D_C \in (1 \pm \delta) \cdot \Diam(C)$.
	Then, for each cell $R' \in \cell(R)$ of each simple square $R$, we examine each point in $N_C$,
	and if $\dist(R', N_C) \leq \delta \cdot D_C$, we add a new point $x \in R'$
	such that $\dist(x, N_C) \leq \delta \cdot D_C$ to the stream,
	and assign it color $C$.
	Furthermore, we declare $C$ intersects $R'$.
	Crucially, when we need to add a new point for some other color $C'$ but for the same cell $R'$,
	we need to add the same point $x \in R'$. This is to ensure the extra cost can be efficiently charged to $\OPT$,
	since otherwise the total added cost for one cell could be as large as $\delta \sum_C \Diam(C)$, while we only have $\OPT \geq \max_C \Diam(C)$ as a lower bound for $\OPT$.

In \cref{lemma:streaming_net}, we present a streaming algorithm to compute the covering $N_C$ mentioned above.
	Roughly speaking, this algorithm makes use of a sparse recovery data structure (\Cref{lemma:streaming_recover}) to find the grid points in squares of side-length $\approx \delta \cdot \Diam(C)$.
	Only the grid points that are within the target distance $\delta \cdot \Diam(C)$ are kept and recovered,
	and this is crucially used to bound the size of the returned covering and thus the space complexity.

	\begin{lemma}
		\label{lemma:streaming_net}
		There is an algorithm that for every $0 < \delta \leq 1$ and
		every point set $S \subset \mathbb{R}^2$ provided as a dynamic geometric stream,
		computes a subset $N_S \subset \mathbb{R}^2$ that is an $O(\delta) \cdot \Diam(S)$-covering for $S$,
with probability at least $1 - \poly(\Delta^{-1})$,
		using space $O(\delta)^{-2} \cdot \polylog \Delta$,
		and running in time $O(\delta)^{-2} \cdot \polylog \Delta$ per stream update.
	\end{lemma}
	\begin{proof}
		We give the procedure in \cref{alg:streaming_net}. It makes use of \cref{lemma:streaming_recover} in a way similar to \cref{alg:min_square}.
		\begin{algorithm}[t]
			\caption{Streaming algorithm for constructing the covering}
			\label{alg:streaming_net}
			\begin{algorithmic}[1]
				\Procedure{Find-Covering}{$S \subseteq [\Delta]^2, \delta \in (0, 1]$}
					\Comment{$S$ presented as a dynamic point stream}
					\State for each $1 \leq i \leq \log_2{L}$, maintain sketch $\mathcal{K}_i$ of \cref{lemma:streaming_recover}
					for the $2^i$-grid $\mathcal{G}_i$, with parameter $T = (3\delta^{-1})^{2}$
					\For{insertion/deletion of point $x$ in the stream}
						\For{$i \gets 0, \ldots, \log_2{L}$}
							\State let $y \in \mathcal{G}_i$ be the grid point for the level-$i$ square that $x$ belongs to
							\State increase/decrease the frequency of $y$ by $1$
							in sketch $\mathcal{K}_i$
						\EndFor
					\EndFor
					\Comment{the stream terminates, and the computing phase starts}
					\State find the least $i$ such that $\mathcal{K}_i$ reports the number of elements is $\leq 2T$
					\State \textbf{return} the $\leq 2T$ grid points in $\mathcal{G}_i$ that $\mathcal{K}_i$ reports
				\EndProcedure
			\end{algorithmic}
		\end{algorithm}
		Observe that the only source of randomness is the sparse recovery structure (\Cref{lemma:streaming_recover}),
		and that the parameters are set so that we can reach the target failure probability bound conditioning on the success of all the sparse recovery structures.
		Below, condition on the success of all the sparse recover structure.

		In \Cref{alg:streaming_net}, for the iterations with $i$ such that $2^i \geq \delta \cdot \Diam(C)$,
		every two $2^i$-grids have distance at least $\delta \cdot \Diam(C)$,
		so by \Cref{fact:packing}, the sparse recovery would find at most $T$ points (conditioning on success).
		Thus, the algorithm would return a set of grid points $W$
		with $|W| \leq 2T$ and that each point $x \in S$ has a point in $W$ within distance $O(\delta) \cdot \Diam(C)$.
		This finishes the proof.
	\end{proof}

	In fact, such procedure of adding points is oblivious to the subproblem,
	and should be done only once as a \emph{pre-processing step} before evaluating any subproblems.
	Therefore, the subproblems are actually evaluated on a new instance
	$(X', \mathcal{C}')$ after the pre-processing.
	Since we apply \cref{lemma:streaming_net} for every color $i$,
	and by the choice of $\delta$, the space complexity for the pre-processing
	step is $O\left(k^3\cdot \poly(\epsilon^{-1} \log \Delta)\right)$,
	and the time complexity per update is bounded by this quantity.
Next, we upper bound the error introduced by the new instance.
	\begin{lemma}
		\label{lemma:error_new_inst}
		Let $\OPT$ be the optimal SFP solution for the original instance $(X, \mathcal{C})$,
		and let $\OPT'$ be that for $(X', \mathcal{C}')$.
		Then $w(\OPT) \leq w(\OPT') \leq (1 + \epsilon) \cdot w(\OPT)$.
	\end{lemma}
	\begin{proof}
		Since $\OPT'$ is a feasible solution for $(X, \mathcal{C})$,
		we obtain $w(\OPT) \leq w(\OPT')$ by the optimality of $\OPT$.
		It remains to prove the other inequality.

		Recall that for every color $C$,
		we use \cref{lemma:streaming_net} to obtain a $\delta \cdot \Diam(C)$-covering $N_C$ and
		estimate $\Diam(C)$ using $D_C := \Diam(N_C)$.
		Now, for every cell $R'$ of every simple square, if $\dist(R', N_C) \leq \delta \cdot D_C$
		for some color $C$
we add a point $x$ to $C$
		satisfying $d(x, N_C) \leq \delta \cdot D_C$,
		and for any other color $C'\neq C$ with $\dist(R', N_{C'}) \leq \delta \cdot D_{C'}$,
		we add the same point $x$ to $C'$. Note that we add at most one distinct point for each cell.
		For an added point $x$, let $z\in N_C$ satisfy $d(x, z)\leq \delta \cdot D_C$.
		Then adding $x$ increases $\OPT$ by at most $2\delta \cdot D_C \leq 3\delta \cdot \Diam(C)$,
		since one can connect $x$ to $y \in C$ such that
		$\dist(y, z) \leq \delta \cdot \Diam(C)$ (such $y$ must exist due to \cref{lemma:streaming_net}).
		
		Since there are in total at most $O(k \log \Delta \cdot \epsilon^{-2})$
		cells in all simple squares by Lemma~\ref{lemma:leaves},
		the total increase of the cost is at most
$O(\delta \cdot k \log \Delta \cdot \epsilon^{-2})
			\cdot \max_{C \in \mathcal{C}}{ \Diam(C) }
			\leq \epsilon \max_{C \in \mathcal{C}}{\Diam(C)}
			\leq \epsilon w(\OPT)$,
using a suitable choice of $\delta = O\left(\epsilon^3 (k\log \Delta)^{-1}\right)$ and recalling that $\max_{C \in \mathcal{C}}(\Diam(C)) \leq w(\OPT)$.
		We conclude that $w(\OPT') \leq (1 + \epsilon) \cdot w(\OPT)$.
\end{proof}

\subsection{Evaluating Basic Subproblems in Streaming}
	\label{sec:streaming_dp}
	After we obtain the new instance $(X', \mathcal{C}')$, we
	evaluate the cost for every subproblem $(R, A, f, \Pi)$.
	Because of the modification of the instance, we know for sure
	the subset of colors $\mathcal{C}_{R'}$ for each cell $R'$.
	To evaluate the subproblem, recall that we start with enumerating a partition of $\cell(R)$ that is compatible
	with the subproblem, which can be tested efficiently using $\mathcal{C}_{R'}$'s.
	Suppose now $\{ P_i := R_i \cup A_i \}_{i=1}^{t}$ is a partition of $\cell(R) \cup A$ that we enumerated, 
where $R_i$ represents the input points in the subset of cells corresponding to the $i$-th part of the partition
(recall that $A$ is the set of active portals, which needs to be connected to cells in a way that is compatible to the constraint $f$).
	Then, as in the offline algorithm, we
	evaluate $\MST(P_i)$ of each part $P_i$, and compute the sum of them, i.e. $\sum_{i=1}^{t}{\MST(P_i)}$, however, we need to show how to do this in the streaming setting.

	Frahling et al.~\cite{FIS08} designed an algorithm that reports
	a $(1+\epsilon)$-approximation for the value of the MST of a point set
	presented in a dynamic stream, using space $O(\epsilon^{-1}\log{\Delta})^{O(1)}$.
	Furthermore, as noted in Section~\ref{sec:intro}, their algorithm maintains a linear sketch.
Now, a natural idea is to apply this MST sketch, that is,
	create an MST sketch for each color, which only takes $k \cdot O(\epsilon^{-1}\log{\Delta})^{O(1)}$ space.
	Then, for each $P_i = R_i \cup A_i$, we compute the set of intersecting colors,
	and we create a new MST sketch $\mathcal{K}$ by first adding up the MST sketches of these colors (recalling that they are linear sketches),
	and then adding the active portals $A_i$ connected to $P_i$ to the sketch.
	We wish to query the sketch $\mathcal{K}$ for the cost of $\MST(P_i)$.

	However, this idea cannot directly work, since the algorithm by \cite{FIS08} only gives
	the MST value for \emph{all} points represented by $\mathcal{K}$, instead of the MST value for a subset $P_i$.
	Therefore, we modify the MST sketch to answer the MST cost
	of a subset of points of interest. 

\medskip

\paragraph{Generalizing the MST algorithm to handle subset queries.}\footnote{
	Here we assume the reader has some familiarity with the MST algorithm of \cite{FIS08}; we provide a more detailed overview of \cite{FIS08} in \Cref{sec:proof_mst_add}.}
	Fix some part $P_i$.
Recall that the $P_i$'s always consist of at most $O(\epsilon^{-2})$ cells
	(which are quadtree squares), plus $O(\epsilon^{-1})$ active portal points.
	Hence, a natural first attempt is to use or modify the $\ell_0$-samplers
        so that they sample only from the squares defined by $P_i$.
	Unfortunately, this approach would not work, since the squares are not known in advance and may be very small (i.e., degenerate to a single point),
	so sampling a point from them essentially solves the INDEX problem.
	
	Therefore, to estimate $\MST(P_i)$, we still use the original $\ell_0$-samplers,
	and we employ a careful sampling and estimation step.
	We sample from the whole point set maintained by the sketch $\mathcal{K}$ but keep only the samples contained in $P_i$.
	We execute the stochastic-stopping BFS from the points that are kept,
        restricting the BFS to points contained in $P_i$.

	One outstanding problem of this sampling method is that if the number of points in $P_i$,
	or to be exact, the number of non-zero entries of level-$i$ $\ell_0$-samplers,
	is only a tiny portion of that of the entire point set,
    then with high probability, we do not sample
	any point from $P_i$ at all.
	In this case, no BFS can be performed,
	and we inevitably answer $0$ for the number of successful BFS's,
	which eventually leads to some overall additive error.
The next lemma summarizes the above discussion and provides a bound on this additive error.

\begin{restatable}{lemma}{lemmaStreamingMSTadd}
\label{lemma:streaming_mst_add}
There is an algorithm that, given $0 < \epsilon < 1/2$, integers $k,\Delta \geq 1$,
and a set of points $S \subseteq [\Delta]^2$ presented as a dynamic geometric stream,
maintains a linear sketch of size $k^2 \cdot \poly(\epsilon^{-1} \log k \log{\Delta})$.
For every query $(R, \{R_j\}_{j=1}^{t}, A)$ at the end of the stream, 
where
\begin{enumerate} 
\item $R$ is a simple square, $A$ is a subset of portals of $R$, and
\item $\{R_j\}_{j=1}^{t} \subseteq \cell(R)$,
\end{enumerate}
the algorithm computes from the linear sketch an estimate $E\ge 0$
such that with probability at least $1 - \exp(-\log k\cdot \poly(\epsilon^{-1} \log\Delta))$,
\begin{align}
  \MST(P)
  \leq E
  \leq (1 + \epsilon) \cdot \MST(P) + O\left(\frac{\poly(\epsilon)}{k \cdot \log \Delta}\right) \cdot \MST(S)\,, \label{eqn:lemmaStreamingMSTadd}
\end{align}
where $P = \left(S \cap  \left(\bigcup_{j=1}^{t}{R_j}\right)  \right) \cup A$.
The algorithm runs in time $k^2 \cdot \poly(\epsilon^{-1} \cdot \log k\cdot \log{\Delta})$ per update,
and the query time is also $k^2 \cdot \poly(\epsilon^{-1} \cdot \log k\cdot \log{\Delta})$.
\end{restatable}

This lemma constitutes the main algorithm for the evaluation of the subproblem.
	Note that we only need to prove it for one point set $S$, since the sketch is linear.
	Indeed, when applying \cref{lemma:streaming_mst_add}, we obtain the sketch for each color separately from the stream,
	and for every query, we first merge the sketches of colors relevant to the query and add query portals to the resulting sketch.
	By linearity, this is the same as if we obtain the sketch for all these colors and portals at once.
We present the proof of \cref{lemma:streaming_mst_add} in \cref{sec:proof_mst_add},
	where we give a more detailed discussion of the technical issues and our novel ideas to overcome them.
	The proof of \cref{thm:main-precise} using \cref{lemma:streaming_mst_add} appears in \cref{sec:proof_kd}.

\subsection{Proof of \cref{thm:main-precise}}
\label{sec:proof_kd}
	
This section proves \cref{thm:main-precise}, which is restated below for convenience.
	
\thmmainprecise*
	
	We combine the above building blocks to prove \cref{thm:main-precise}.
	A description of the complete algorithm can be found in \cref{alg:streaming_kd};
	here, we provide a brief outline.
		When processing the stream, we maintain a sketch of \cref{lemma:streaming_simp_square} to find the maximal simple squares $\mathcal{R}$,
		a set of sketches of \cref{lemma:streaming_net}
		for every color to construct the covering for approximate compatibility checking,
		and a set of linear sketches of \cref{lemma:streaming_mst_add} for every color to answer MST cost queries.
		Upon a query to estimate the SFP cost,
		we compute the simple squares $\mathcal{R}$ by simulating \cref{alg:simpleSquares}
		and use the coverings to add artificial points to every cell of a square in $\mathcal{R}$ that is close enough to a point in one of the coverings
		as described in Section~\ref{sec:streaming_comp_check};
		the latter allows for checking the compatibility of a partition of cells with the color constraints.
		The main part is to estimate the values of basic DP subproblems, which we enumerate
		for every simple square in $\mathcal{R}$. For each such subproblem, we enumerate all
		partitions of cells with at most $O(\eps^{-1})$ parts (which is sufficient by \Cref{thm:dp_struct})
		and check compatibility of the partition with the subproblem and the color constraints as described in \Cref{sec:offline}.
		Finally, for each local component of each compatible partition, we query the merge of sketches of \cref{lemma:streaming_mst_add}
		for the relevant colors (i.e., those appearing in the component) to get an estimate of the MST cost of the local component.

	\begin{algorithm}[H]
	\caption{Main streaming algorithm}
	\label{alg:streaming_kd}
	\begin{algorithmic}[1]
	\Procedure{SFPinitialization}{$\mathcal{C}$} \Comment{$\mathcal{C}$ is the set of colors}
		\State initialize a sketch $\mathcal{K}^{(1)}$ of \cref{lemma:streaming_simp_square},
		a set of sketches of \cref{lemma:streaming_net}
		$\{ \mathcal{K}^{(2)}_C \}_{C \in \mathcal{C}}$
		for every color $C \in \mathcal{C}$ with parameter $\delta := \poly(\epsilon)(k\log \Delta)^{-1}$,
		and a set of (linear) sketches\footnotemark\  
		of \cref{lemma:streaming_mst_add}
		$\{ \mathcal{K}^{(3)}_C \}_{C \in \mathcal{C}}$
		for every color $C \in \mathcal{C}$
	\EndProcedure
\medskip
	\Procedure{SFPupdate}{$x, C$, \texttt{insert/delete}}
		\Comment{insert/delete point $x$ of color $C$}
		\State insert/delete point $x$ in sketches $\mathcal{K}^{(1)}, \mathcal{K}^{(2)}_C, \mathcal{K}^{(3)}_C$
	\EndProcedure
		\medskip
		
	\Procedure{SFPquery}{} \Comment{the stream terminates}
		\State use sketch $\mathcal{K}^{(1)}$ to compute a set of simple squares $\mathcal{R}$
		\Comment{see \cref{sec:streaming_simp_square}}
\State for each color $C \in \mathcal{C}$, use sketch $\mathcal{K}^{(2)}_C$
		to compute a covering $N_C$, and let $D_C := \Diam(N_C)$
		\Comment{$D_C$ is a $(1\pm \epsilon)$-approximation for $\Diam(C)$}

\State initialize a Boolean list $\mathcal{I}$ that records whether
		a cell of a simple square and a color intersect

		\Comment{This uses space at most $O(k\cdot \log \Delta\cdot \poly(\epsilon^{-1}))$}
		\For{every $R \in \mathcal{R}$, $R' \in \cell(R)$} \If{$\dist(N_C, R') \leq \delta \cdot D_C$ for some color $C$}
		\State let $x \in R'$ be a point such that $\dist(x, N_C) \leq \delta \cdot D_C$
		\For{every color $C'$ with $\dist(N_{C'}, R') \leq \delta \cdot D_{C'}$}
		\State add $x$ to $\mathcal{K}^{(3)}_{C'}$ and record in $\mathcal{I}$ that $R'$ intersects color $C'$
		\Comment{see \cref{sec:streaming_comp_check}}
		\EndFor
		\EndIf
		\EndFor
		\For{each simple square $R$ and an associated subproblem $(R, A, f, \Pi)$} \label{line:start_subproblem}
		\For{each partition of $\cell(R)$ with at most $O(\eps^{-1})$ parts}
		\If{the partition is compatible with the subproblem} \Comment{see \cref{sec:streaming_comp_check}}
		\For{each part $\{R_j\}_{j=1}^{t} \subseteq \cell(R)$ in the partition}
		\State let $A_j \subseteq A$ be the set of active portals that $\{R_j\}_{j=1}^{t}$ connects to
		\State create linear sketch $\mathcal{K}'$,
		by adding up $\mathcal{K}^{(3)}_C$ for every $C$
		intersecting a cell $R_j$ for $j \in \{1, \dots, t\}$
		
		\Comment{the intersection information is recorded in $\mathcal{I}$}
		\State add points in $A_j$ to sketch $\mathcal{K}'$
		\State query sketch $\mathcal{K}'$ for the value of the MST of the part $\{R_j\}_{j=1}^{t}$ and portals $A_j$ (as in \cref{lemma:streaming_mst_add}) \label{line:query_subproblem}
		\Comment{see \cref{sec:streaming_dp}}
		\EndFor
		\State store the sum of the queried values of $\MST(\{R_j\}_{j=1}^{t}, A_j)$ as the estimated cost for the subproblem
		\EndIf
		\EndFor
		\EndFor \label{line:end_subproblem}
		\State invoke the DP (as in~\cite{DBLP:journals/algorithmica/BateniH12}) using the values of basic subproblem estimated as above
		\State \textbf{return} the DP value (for the root square with no active portals)
		\EndProcedure
	\end{algorithmic}
\end{algorithm}
\footnotetext{We need to use the same randomness for sketches $\{ \mathcal{K}_C^{(3)} \}$ among all colors $C$ so that they can be combined later.} 	
	
	The space and update time follow immediately
	from \cref{alg:streaming_kd}, \cref{thm:dp_struct,lemma:streaming_simp_square,lemma:streaming_net,lemma:streaming_mst_add}.

	To bound the query time, note that there are at most
	\begin{itemize}[nosep]
		\item $O(k\cdot \log \Delta)$ simple squares in $\mathcal{R}$,
		and thus, $O(k\cdot \log \Delta)$ quadtree squares for which we evaluate DP subproblems,
		\item $(\epsilon^{-1}\cdot \log \Delta)^{O(\epsilon^{-2})}$ DP subproblems associated with each square (see Section~\ref{sec:dp_review}), and
		\item $\epsilon^{-O(\epsilon^{-1})}$ MST queries evaluated for each subproblem,
			since there are most $O(\eps^{-1})$ MST queries per each of the at most
			$\epsilon^{-O(\epsilon^{-1})}$ partitions of cells with at most $O(\eps^{-1})$ parts.
	\end{itemize}
	Furthermore, the time complexity of an MST query is $k^2 \cdot \poly(\log k\cdot\epsilon^{-1} \log{\Delta})$ by \cref{lemma:streaming_mst_add}, so
	the query time is bounded by
$$O(k\cdot \log \Delta)\cdot (\epsilon^{-1}\cdot \log \Delta)^{O(\epsilon^{-2})}\cdot \epsilon^{-O(\epsilon^{-1})} \cdot k^2 \cdot \poly(\log k\cdot\epsilon^{-1} \log{\Delta}) \le k^3\cdot \poly(\log k)\cdot (\epsilon^{-1}\log\Delta)^{O(\epsilon^{-2})}\,.$$

	To bound the failure probability,
	we use a union bound over the failure probabilities of all applications
	and queries of the streaming algorithms as well as the error bound in \cref{thm:dp_struct}.
	We observe that \cref{thm:dp_struct} incurs an $O(1)$ failure probability,
	and every other step, except for the use of \cref{lemma:streaming_mst_add},
	have a failure probability of $\poly(\Delta^{-1})$.
	To bound the failure probability of all applications of \cref{lemma:streaming_mst_add},
	observe that we have $k\cdot (\epsilon^{-1}\cdot \log \Delta)^{O(\epsilon^{-2})}$ basic subproblems (see \cref{sec:prelim}),
	and for each basic subproblem we need to evaluate at most $\epsilon^{-O(\epsilon^{-1})}$ MST queries.
	Thus, the total failure probability of evaluating the subproblems is at most
	\begin{align*}
		k\cdot (\epsilon^{-1}\cdot \log \Delta)^{O(\epsilon^{-2})}\cdot \epsilon^{-O(\epsilon^{-1})} \cdot \exp(-\log k\cdot \poly(\epsilon^{-1}\log \Delta))
		\leq \poly(\Delta^{-1})\,,
	\end{align*}
	by the guarantee of \cref{lemma:streaming_mst_add}.
Therefore, we conclude that the failure probability is at most
	$\frac{1}{3}$. It remains to analyze the error.

	\paragraph{Error analysis.}
	For the remaining part of the analysis, we condition on no failure of the sketches used in \cref{alg:streaming_kd}
	and on that the error bound in \cref{thm:dp_struct} holds.
	By \cref{lemma:error_new_inst}, for the part of evaluating the
	basic subproblems (lines \ref{line:start_subproblem}-\ref{line:end_subproblem} of \cref{alg:streaming_kd}),
	the actual instance that the linear sketches
	work on is $(1 + O(\epsilon))$-approximate.
	Hence, it suffices to show the DP value is accurate to that instance.

	Our estimation is never an underestimate, by \cref{lemma:streaming_mst_add}
	and since all partitions that we enumerated are compatible with the subproblems;
	see \cref{sec:streaming_comp_check}.
	Hence, it remains to upper bound the estimation.
	Consider an optimal DP solution $F$, which we interpret as a metric graph (see \cref{sec:prelim}).
	Then we create a new solution $F'$ from $F$ by modifying $F$ using the following procedure.
	For each simple square $R$, we consider $F_{|R}$ which is the portion of
	$F$ that is totally inside of $R$ (see \cref{sec:prelim}).
	For each component $S \subset R$ in $F_{|R}$, let $S'$ be the point set formed by removing all Steiner points from $S$,
	except for portals of $R$ (note that we remove portals of subsquares of $R$ if they appear in $S$).
	Then, for each component $S$, we replace the subtree in $F$ that spans $S$
	with the MST on $S'$.
	It is immediate that after the replacement, the new solution has the same
	connectivity of portals and terminal points as before.
	We define $F'$ as the solution after doing this replacement for all simple squares.

	$F'$ is still a feasible solution. Furthermore, for every simple square $R$,
	if $F$ is compatible with a subproblem $(R, A, f, \Pi)$, then
	so does $F'$.
	By the construction of $F'$, the definition of Steiner ratio $\alpha_2$,
	and \cref{thm:dp_struct},
	we know that
\begin{equation}
		w(F') \leq \alpha_2 \cdot w(F) \leq (1 + O(\epsilon)) \cdot \alpha_2 \cdot \OPT, \label{eqn:f_fprime}
	\end{equation}
	where the last inequality holds as we condition on that the error bound  in \cref{thm:dp_struct} holds.

	Now we relate the algorithm's cost to $w(F')$.
	Fix a simple square $R$, and suppose $(R, A, f, \Pi)$ is the subproblem
	that is compatible with $F'_{|R}$.
	Then, the components in $F'_{|R}$ can be described by a partition
	of the cells plus their connectivity to active portals.
	Such a subproblem, together with the partition, must be examined
	by \cref{alg:streaming_kd} (in lines \ref{line:start_subproblem}-\ref{line:end_subproblem}),
	and the MST value for each part is estimated
	in line~\ref{line:query_subproblem}.
	Since the algorithm runs a DP using the estimated values,
	the final DP value is no worse than the DP value that is only evaluated from
	the subproblems that are compatible to $F'$.
	Recall that our estimation for each subproblem not only has a multiplicative error
	of $(1+\epsilon)$ but also an additive error by \cref{lemma:streaming_mst_add}.
	By the fact that $F'$ always uses MST to connect points in components of basic subproblems,
	it suffices to bound the \emph{total} additive error for the estimation
	of the MST cost of the components of $F'$. 

	Fix a connected (global) component $Q$ of $F'$, and let $\mathcal{C}_Q \subseteq \mathcal{C}$
	be the subset of colors that belongs to $Q$.
	By \cref{lemma:streaming_mst_add}, for every basic subproblem $(R, f, A, \Pi)$ that is compatible with $F'$,
	and every component $P$ of $Q_{|R}$, the additive error
	is at most
		$O\left(\frac{\poly(\epsilon)}{k \log \Delta}\right) \cdot \MST(S)$,
	where $S$ is the union of color classes that intersect $P$ plus the active portals $A$.
	Observe that $\mathcal{C}_S\subseteq \mathcal{C}_Q$ (where $\mathcal{C}_S$ is the set of colors used in $S$), so
	$S$ is a subset of the point set of $Q$ (note that $Q$ contains all portals in $A$ as $F'$ is a portal-respecting solution and the subproblem is compatible with $F'$)
	and thus $\MST(S)\le \MST(Q)$, which implies
	\begin{equation*}
		O\left(\frac{\poly(\epsilon)}{k \log \Delta}\right) \cdot \MST(S)
		\leq O\left(\frac{\poly(\epsilon)}{k \log \Delta}\right)\cdot \MST(Q)
		\leq O\left(\frac{\poly(\epsilon)}{k \log \Delta}\right) \cdot w(Q).
\end{equation*}
	Observe that for each simple square $R$, $Q_{|R}$ has at most $O(\epsilon^{-1})$ local components by \cref{thm:dp_struct}.
	Hence, summing over all local components of $Q_{|R}$ and all simple squares $R$, the total additive error is bounded by
$O(\epsilon^{-1}) \cdot O(k \log \Delta) \cdot O\left(\poly(\epsilon) / (k \log \Delta)\right) \cdot w(Q)
		\leq \epsilon \cdot w(Q)$, where we used that there are at most $O(k \log \Delta)$ simple squares by Lemma~\ref{lemma:leaves}.
	Finally, summing over all components $Q$ of $F'$, we conclude that the total additive error is $\epsilon \cdot w(F')$,
	and combined with \cref{eqn:f_fprime}, the error guarantee follows.
	This finishes the proof of \cref{thm:main-precise}.

\section{Proof of Lemma~\ref{lemma:streaming_mst_add}: Subset MST Query in Streaming}
\label{sec:proof_mst_add}

This section proves \cref{lemma:streaming_mst_add}, which is restated below for convenience.

\lemmaStreamingMSTadd*

	The general proof strategy is similar to that of~\cite{FIS08},
	hence we start with a review of~\cite{FIS08},
        with minor adjustments to our setting,
	described in terms of a generic input $V \subseteq [\Delta]^2$. 

    \begin{table}[ht]
        \caption{List of major notations in the proof of \Cref{lemma:streaming_mst_add}}
        \label{tab:notation}
        \centering
        \begin{tabular}{ll}
            \toprule
            notation & definition \\
            \midrule
            $S$ & input dataset \\
            $P$ & the subset in the query \\
            $V$ & a generic set of points (used mostly in definitions and lemma statements) \\
            $G_V$ & the (complete) metric graph on point set $V$ \\
            $G_V^{(i)}$ & $(1 + \epsilon)^i$-threshold graph of $G_V$ \\
            $c_V^{(i)}$ & number of connected components in $G_V^{(i)}$ \\
            $i'$ & $i' := \log_2 O(\epsilon \cdot (1 + \epsilon)^i)$, as in \eqref{eqn:iprime} \\
            $\mathring{G}_V^{(i)}$ & rounded $(1 + \epsilon)^i$-threshold graph \\
            $\mathring{c}_V^{(i)}$ & number of components in $\mathring{G}_V^{(i)}$ \\
            $\mathring{n}_V^{(i)}$ & number of (distinct) vertices in $\mathring{G}_V^{(i)}$ \\
            parameters $\kappa, \sigma, \Gamma$ & \eqref{eqn:param_def} \\
            $\widetilde{n}_V^{(i)}$, $\widetilde{c}_V^{(i)}$ & estimates for $\mathring{n}_V^{(i)}$ and $\mathring{c}_V^{(i)}$ defined in \Cref{alg:mst_add_query} \\
            parameters $\lambda_1, \lambda_2$ & $\lambda_1 := \frac{\kappa}{2\sigma}$, $\lambda_2 := \frac{2\kappa}{\sigma}$ \\
            \bottomrule
        \end{tabular}
    \end{table}

	\paragraph{Review of~\cite{FIS08}.}
	A key observation is that the cost of the MST of a point set $V \subseteq [\Delta]^2$
	can be related to the number of components in the metric threshold graphs of different scales;
	a similar observation was first given in~\cite{DBLP:journals/siamcomp/ChazelleRT05}.
	In particular, let $G_V$ be the complete metric graph on $V$,
	i.e., the vertex set is $V$, the edge set is $\{\{ u, v\} : u\neq v \in V\}$, and the edge weights are given by $\dist(\cdot, \cdot)$.
For $i \geq 0$, let $G_V^{(i)}$ be the $(1 + \epsilon)^i$-threshold graph,
	which consists only of the edges of $G_V$ that have weight at most $(1 + \epsilon)^i$.
	Let $c_V^{(i)}$ be the number of connected components in $G_V^{(i)}$.
	Then for every $W \geq \Diam(V)$ that is a power of $(1 + \epsilon)$,
	\begin{align}
		\MST(V)
		\leq c_V^{(0)} - W + \epsilon \cdot \sum_{i=0}^{\log_{1+\epsilon}{W} - 1}{ (1 + \epsilon)^{i} \cdot c_V^{(i)} }
		\leq (1 + \epsilon) \cdot \MST(V).
		\label{eqn:mst_component}
	\end{align}

	It remains to estimate $c_V^{(i)}$ for $0 \leq i \leq \log_{1 + \epsilon}{W} - 1$.
For each $i$, Frahling et al.~\cite{FIS08} consider the subdivision of $[\Delta]^2$
	into squares of sidelength $(1+\epsilon)^{i'}$ for a small enough $i'$,
	namely, $i'$ is the largest integer satisfying $(1+\epsilon)^{i'} \leq O(\epsilon \cdot (1+\epsilon)^i)$.
        They round the input points to grid points formed by centers of these squares.
	However, this is not convenient in our setting as we need to restrict the MST query to a subset of quadtree squares
	(with sidelengths of powers of~$2$). As it is not possible to ``align'' squares of size $(1+\epsilon)^{i'}$ with quadtree squares,
	the aforementioned rounding could move a point not relevant to the MST query to a quadtree square of the query
	(thus making it appear relevant), or vice versa.
	
	To avoid this issue, we adjust the rounding for MST cost estimation so that the grid points are centers
	of squares in the randomly-shifted quadtree we use in the DP computation.
	Recall from Section~\ref{sec:prelim} that $\mathcal{G}_{i'}$ is
	the set of centers of all level-$i'$ quadtree squares.
	Namely, to estimate $c_V^{(i)}$, we round input points to the $2^{i'}$-grid $\mathcal{G}_{i'}$,
	where $i'$ is the largest integer such that $2^{i'} \leq O(\epsilon \cdot (1+\epsilon)^i)$, that is,
	\begin{equation}\label{eqn:iprime}
	i' := \log_2 O(\epsilon \cdot (1+\epsilon)^i)\,.
	\end{equation}
	We require the constant hidden in the $O$ notation to be sufficiently small so that an inequality in~\eqref{eqn:relevanceTesting}
	holds.
	(If $i' < 0$ according to this definition, then we can of course take $i' = 0$ and no rounding is needed as we reach the granularity
	of the input data.)
	Our rounding is only finer compared to~\cite{FIS08} (i.e., to centers of smaller squares),
	but within a constant factor, so the space bound remains asymptotically the same.
	Note that while the squares of $\mathcal{G}_{i'}$ are from the quadtree, we still need to show
	that we do not round to centers of larger squares than the cells of the MST query, which will
	imply that indeed, our rounding does not move a point irrelevant for the query to a cell of the query, or vice versa.
	We remark that for several consecutive indexes $i\in [0, \log_{1 + \epsilon}{W} - 1]$, the index $i'$ could be of the same value.

	Similarly to~\cite{FIS08},
	define for each $i$ a ``rounded'' metric graph $\mathring{G}^{(i)}_V$ as follows.
	\begin{enumerate}
		\item Move each point $x \in V$ to the center $y \in \mathcal{G}_{i'}$ of the quadtree square that $x$ belongs to,
		where $i'$ is defined as in~\eqref{eqn:iprime}.
		\item The vertex set of $\mathring{G}^{(i)}_V$ consists of
		\emph{non-empty} grid points in $\mathcal{G}_{i'}$, i.e., for each vertex $v$ at least one point in $V$ was
		moved/rounded to the grid point corresponding to $v$, and the edge set consists of pairs of vertices that
		are of distance at most $(1 + \epsilon)^i$.
	\end{enumerate}
	We refer to the vertices of $\mathring{G}^{(i)}_V$ as \emph{non-empty} grid points.
	Let $\mathring{c}^{(i)}_V$ be the number of components in $\mathring{G}^{(i)}_V$,
	and let $\mathring{n}^{(i)}$ be the number of vertices in $\mathring{G}^{(i)}_V$.
	As shown in~\cite{FIS08}, $c^{(i)}_V$ is well approximated by $\mathring{c}^{(i)}_V$,
	which we restate in \cref{lemma:num_component_relation}.
	Strictly speaking, \cref{lemma:num_component_relation} (\cite[Claim 4.1]{FIS08}) is for
	the rounding to centers of squares with sidelength $O(\epsilon\cdot (1 + \epsilon)^i)$,
	but its proof more generally applies to any rounding which moves points by at most $O(\epsilon\cdot (1 + \epsilon)^i)$
	and thus, to our rounding as well.
	\begin{lemma}[{\cite[Claim 4.1]{FIS08}}]
		\label{lemma:num_component_relation}
		For every $i$, $c_V^{(i+1)} \leq \mathring{c}^{(i)}_V \leq c^{(i-2)}_V$.
	\end{lemma}
	Therefore, we focus on estimating $\mathring{c}^{(i)}_V$'s.
	As in~\cite{FIS08}, it suffices to account for \emph{small}
	components that have at most $O(\epsilon^{-2} \log\Delta)$ non-empty grid points.
	The reason is that the number of large components that have more than
	$\epsilon^{-2}\log \Delta$ points is at most $O(\epsilon^2 / \log\Delta \cdot \mathring{n}^{(i)}_V)$,
	which contributes $O(\epsilon^2  / \log\Delta \cdot (1 + \epsilon)^i \cdot   \mathring{n}^{(i)}_V)$ in \cref{eqn:mst_component}.
	This contribution can be bounded using the following lemma.\footnote{Note that Lemma~4 in~\cite{FIS08}
          proves the lower bound only for the case when $W\ge \Omega(\epsilon\cdot (1+\epsilon)^i)$,
          however, using the argument of the first case of their proof implies our bound.
        }
	
	\begin{lemma}[Lower bound on MST~{\cite[Lemma 4]{FIS08}}]
		\label{mst:lb}
		For every $i \geq 0$ and any point set $V\subseteq S$ with $\mathring{n}^{(i)}_V \ge 8$,\footnote{The constant $8$ is not arbitrarily picked --- in general dimension $d$, this bound is $\mathring{n}_V^{(i)} \geq 2^{d+1}$.} it holds that $\MST(V) \geq \Omega(\epsilon\cdot (1+\epsilon)^i \cdot \mathring{n}^{(i)}_V)$.
	\end{lemma}

	To estimate the number of small components, we use the BFS algorithm with a stochastic stopping condition;
	this idea was first applied in~\cite{DBLP:journals/siamcomp/ChazelleRT05}.
	In particular, for each $i$,
	$\poly(\epsilon^{-1}\log \Delta)$ samples are taken from the point set of
	$\mathring{G}^{(i)}_V$, which may be efficiently maintained and sampled using $\ell_0$-samplers.
	After that, we perform a stochastic-stopping BFS in $\mathring{G}^{(i)}_V$
	starting from each sampled point and using a random number of steps (but at most $O(\epsilon^{-2} \log \Delta)$ steps).
	An estimate for the number of small components, and thus for $\mathring{c}^{(i)}_V$, is computed using the outcome of each BFS,
	i.e., whether the whole component is discovered or not.
	The random exploration is made possible by the following modified $\ell_0$-sampler designed in~\cite{FIS08},
	which also returns the non-empty neighborhood when sampling a point.
	See also a survey about $\ell_0$-samplers by Cormode and Firmani~\cite{DBLP:journals/dpd/CormodeF14}.

\begin{lemma}[$\ell_0$-Sampler with neighborhood information~{\cite[Corollary 3]{FIS08}}]
\label{lemma:neighbor_l0}
There is an algorithm that, given $\delta > 0$, integers $\rho,\Delta\ge 1$,
and a set of points $S\subseteq [\Delta]^2$ presented as a dynamic geometric stream,
succeeds with probability at least $1-\delta$ and, conditioned on it succeeding,
returns a point $p\in S$ such that for every $s\in S$ it holds that $\Pr[p = s] = 1 / |S|$.
Moreover, if the algorithm succeeds, it returns also all points $s\in S$ such that $\dist(p, s)\le \rho$.
The algorithm has space and both update and query times
bounded by $\poly(\rho\cdot \log\Delta\cdot \log \delta^{-1})$,
and its memory contents is a linear sketch of $S$.
\end{lemma}

	\paragraph{Handling subset queries.}
	In our case, we need to apply \cref{eqn:mst_component}
	with $V = P$, recalling that $P$ is the point set of the query.
	To pick $W$ in~\eqref{eqn:mst_component},
	we need an upper bound on $\Diam(P)$.
	Suppose in the query, the simple square $R$ is of level-$i_0$,
	then $\Diam(P) \leq 2^{i_0}$. Hence, we pick $W$ to be the smallest power of $(1+\epsilon)$ that is no smaller than $2^{i_0}$,
	which implies
	\begin{equation}\label{eqn:diameter_bound}
	2^{i_0} \leq W \leq (1+\epsilon) \cdot 2^{i_0}\,.
	\end{equation}		

	However,
	the query set $P$ is given after the stream, and it is provided
	in a compact form as a union of several cells and portals.
	On the other hand, what we can maintain is only a sketch for $S$ (the whole point set).
	Therefore, the above idea from~\cite{FIS08} cannot immediately solve our problem,
	and as discussed in \cref{sec:streaming_dp},
	we cannot count on adapting $\ell_0$-samplers to directly
	work on a subset unknown in advance either.
	Hence, we still have to maintain the sketch for the whole $S$,
	and query the $\ell_0$-samplers to obtain a sample from $S$.
	Then naturally, the only way out seems to
	filter out the irrelevant sampled non-empty grid points,
	which are those not corresponding to the query set,
	and then we simulate the MST algorithm on the local query instance,
	by using those relevant samples.
	Apart from filtering out irrelevant samples, we also need to restrict the neighborhoods (from \cref{lemma:neighbor_l0})
	to the query cells, i.e., filter out irrelevant points from the neighborhoods.

At a first glance, this idea makes sense and should generally work.
	However, we now discuss an outstanding issue
	and present our new technical ideas for resolving it.

	\paragraph{Additive error and failure probability.}
	As mentioned in \cref{sec:streaming_dp}, the sampling idea may not work
	if $P$ has relatively few relevant grid points $\mathcal{G}_{i'}$, since it is difficult to collect enough relevant grid points using a small number of samples, and in this case
	we suffer an additive error.
	Luckily, we show that if the number of relevant grid points is indeed small relative to that of $S$,
	then the contribution of these points
	to the global MST cost is small after all; see \cref{lemma:np_smaller_ns}.
	This yields a small additive error that can be eventually charged to the global MST of $S$.
	A technical issue here is that the algorithm does not know whether or not
	$P$ contains very few relevant grid points in advance.
	Hence, we introduce an additional sampling step
	to estimate this density, and we use this estimate ($\mathcal{S}^{(i)}$ in \cref{alg:mst_add_query}) to guide the algorithm.
	This works well for the algorithm, but in the analysis,
	the random estimate cannot assert for sure whether or not $P$ contains relatively few relevant grid points.
	We thus need to do the case analysis based on the actual number of grid points in $P$,
	which is not directly aligned with the case separation of the algorithm.

	Finally, since we need to apply the query on all subproblems (in \cref{sec:proof_kd}),
	this requires that both the success probability and the additive error
	are as small as $\frac{1}{k}$ (ignoring other factors).
	In the original analysis by~\cite{FIS08},
	Chebyshev's inequality was used to bound the failure probability,
	which in our case only suffices for an $O(k^3)$ space bound (as opposed to $O(k^2)$ bound in \Cref{lemma:streaming_mst_add}).
	We instead apply Hoeffding's inequality,
	and save a factor of $k$ in space compared to the original calculation in~\cite{FIS08}.

	\paragraph{Algorithm description.}
	We state the complete algorithm for maintaining the linear sketch in \cref{alg:mst_add_sketch},
	and the query algorithm in \cref{alg:mst_add_query}.
	In a nutshell, maintaining the sketch in \cref{alg:mst_add_sketch} requires updating extended $\ell_0$-samplers of \cref{lemma:neighbor_l0} (including the associated neighborhood information) together with keeping track of the $\ell_0$ norm for each level $0, \dots, \log_2{\Delta}$
and a sparse recovery data structure of Lemma~\ref{lemma:streaming_recover} for dealing with the case of a small number of points at level $i$. More precisely, for each threshold $i = 1, \dots, \log_{1+\epsilon}\Delta$,
	we maintain the $\ell_0$-samplers, the sparse recovery algorithm,
and the $\ell_0$ norm in level $i'$ of the quadtree, where $i'$ is defined as in~(\ref{eqn:iprime}).
	Note that \cref{alg:mst_add_sketch} indeed outputs a linear sketch, since the extended $\ell_0$-sampler
	of \cref{lemma:neighbor_l0} is a linear sketch as well as the sparse recovery algorithm and the $\ell_0$-norm estimators (cf.~\cite{DBLP:conf/pods/KaneNW10}).

	In the query procedure (\cref{alg:mst_add_query}), for each threshold $i = 1, \dots, \log_{1+\epsilon}{\Delta}$,
we start by checking
whether or not there are at most $\sigma$ non-empty
	grid points at level $i$ (including those not relevant for the query), and if so, we recover them using the sparse recovery algorithm
	and compute $\mathring{c}^{(i)}_P$, the number of level-$i$ components relevant to the query, without any error.
	Otherwise, we cannot recover all non-empty grid points at level $i$ in a small space.
Then we check whether or not the query contains relatively few non-empty grid points, for which we use the first $\sigma$
	of $\ell_0$-samplers; namely, we check if at least $\kappa$ of sampled grid points are relevant for the query, and if not,
	our estimate for the number of components on that level is simply~$0$.
	Otherwise, the query contains relatively large number of non-empty grid points with high probability,
	and we query the remaining $\sigma$ of $\ell_0$-samplers,
	execute the stochastic-stopping BFS from each relevant sampled grid point,
	and use the outcomes of the BFS (i.e., whether or not the whole component was discovered)
	to estimate the number of components on that level, as described above.
		
	\begin{algorithm}[ht]
		\caption{Algorithm for maintaining sketches for \cref{lemma:streaming_mst_add}}
		\label{alg:mst_add_sketch}
		\begin{algorithmic}[1]
			\Procedure{MST-Sketch}{$S$} \Comment{$S$ is provided as a dynamic geometric stream}
				\State let $\kappa \gets k \log k\cdot\poly(\epsilon^{-1}\log{\Delta})$,
				$\sigma \gets k \poly(\epsilon^{-1}\log{\Delta}) \cdot \kappa$,
				$\Gamma \gets \epsilon^{-3} \log{\Delta}$ \Comment{as in~\eqref{eqn:param_def}}
				\For{$i \gets 0, \ldots, \log_{1 + \epsilon}{\Delta}$
				}
					\State let $i'$ be defined as in~\eqref{eqn:iprime}, i.e., the largest integer such that
						$2^{i'} \leq O(\epsilon \cdot (1 + \epsilon)^i)$
					\State initialize $2 \sigma$
					extended $\ell_0$-samplers $\{\mathcal{K}^{(i)}_j\}_j$ of \cref{lemma:neighbor_l0}
					with frequency vector indexed by $\mathcal{G}_{i'}$ and with neighborhoods containing all
					grid points from $\mathcal{G}_{i'}$ at distance at most $\Gamma\cdot (1+\epsilon)^i$
					\State initialize the sparse recovery algorithm $\mathcal{A}^{(i)}$ of Lemma~\ref{lemma:streaming_recover}
					with $T = \sigma$ for the same frequency vector indexed by $\mathcal{G}_{i'}$

\State initialize an $\ell_0$-norm estimator $\mathcal{N}^{(i)}$ (cf.~\cite{DBLP:conf/pods/KaneNW10})
					for the number of non-empty grid points in $\mathcal{G}_{i'}$ with error parameter $\epsilon$
				\EndFor
				\For{each insertion/deletion of point $x$}
					\For{$i \gets 0, \ldots, \log_{1+\epsilon}{\Delta}$}
						\State let $i'$ be defined as in~\eqref{eqn:iprime}
						\State let $y \in \mathcal{G}_{i'}$
						be the center of the level-$i'$ quadtree square containing $x$ \State increase/decrease the frequency of $y$ by $1$ for estimator $\mathcal{N}^{(i)}$
						\State increase/decrease the frequency of $y$ by $1$ in 
						the sparse recovery $\mathcal{A}^{(i)}$
						\For{each $\ell_0$-sampler $\{\mathcal{K}^{(i)}_j\}_j$} 
							\State increase/decrease by $1$ the frequency of $y$
							\State increase/decrease by $1$ the frequency of $y'$ in the neighborhood of any grid point $z\in \mathcal{G}_{i'}$
								with $\dist(z, y)\le \Gamma \cdot (1 + \epsilon)^i$ (cf.~\cite{FIS08})
						\EndFor
\EndFor
				\EndFor \Comment{stream of $S$ terminates}
			\EndProcedure
		\end{algorithmic}
	\end{algorithm}
        
	Both \cref{alg:mst_add_sketch,alg:mst_add_query} use the following parameters
	\begin{align}\label{eqn:param_def}
		\kappa = k \log k\cdot\poly(\epsilon^{-1} \log \Delta), \quad
		\sigma = k \poly(\epsilon^{-1} \log \Delta) \cdot \kappa, \quad
		\Gamma = \epsilon^{-3} \log{\Delta}.
	\end{align}
	We additionally require that
	\begin{equation}\label{eqn:kappa_sigma_relation}
	\frac{\kappa^2}{\sigma\cdot \Gamma^2} \ge \log k\cdot \poly(\epsilon^{-1} \log \Delta)\,,
	\end{equation}
	where $\poly(\epsilon^{-1} \log \Delta)$ in the RHS is the same as in the failure probability of \cref{lemma:streaming_mst_add};
	this inequality holds for a large-enough $\kappa$.
	Note that $\sigma = k^2 \log k\cdot\poly(\epsilon^{-1} \log \Delta)$.

	\begin{algorithm}[!ht]
		\caption{Algorithm for answering queries of \cref{lemma:streaming_mst_add}}
		\label{alg:mst_add_query}
		\begin{algorithmic}[1]
			\Procedure{Query}{$R, \{ R_1, \ldots, R_t \}, A$}
				\Comment{assume the access to sketches in \cref{alg:mst_add_sketch}}
\State let $\kappa, \sigma, \Gamma$
				be the same parameters as in \cref{alg:mst_add_sketch}
				\State let $W$ be as in~\eqref{eqn:diameter_bound}, i.e., the smallest power of $(1+\epsilon)$ that is no smaller than $2^{i_0}$, where $i_0$ is the level of $R$ in the quadtree
				\For{$i \gets 0, \ldots, \log_{1+\epsilon}{W}$}
					\State let $i'$ be as in~\eqref{eqn:iprime}, i.e., the largest integer such that
						$2^{i'} \leq O(\epsilon \cdot (1 + \epsilon)^i)$ \State query $\mathcal{N}^{(i)}$ and record the value as an estimate $\widetilde{n}^{(i)}_S$ \label{line:mathcaln}
					\State query the sparse recovery algorithm $\mathcal{A}^{(i)}$
					\If{$\mathcal{A}^{(i)}$ answers YES}
\State let $\mathcal{S}^{(i)}$ be the set of 
						at most $\sigma$ non-empty grid points in $\mathcal{G}_{i'}$
						returned by $\mathcal{A}^{(i)}$
						\State compute the number $\mathring{c}^{(i)}_P$ of 
						components in $\mathring{G}^{(i)}_P$ exactly
						\State define estimator $\widetilde{c}^{(i)}_P \gets \mathring{c}^{(i)}_P$
					\Else \Comment{More than $\sigma$ non-empty grid points at level $i$}
					\State query the first $\sigma$ of $\ell_0$-samplers $\{\mathcal{K}^{(i)}_{j}\}_j$,
					and suppose the set of uniformly sampled points is $\{ p_j\}_{j=1}^{\sigma}$ and
					$\{ U_j) \}_{j=1}^{\sigma}$ are the associated neighborhoods in grid $\mathcal{G}_{i'}$
\State let $\mathcal{S}^{(i)} \gets \{ p_j :
						\text{\textsc{Is-Relevant}($i'$, $p_j$ )} = \text{TRUE}  \}$
					be the subset relevant to the query \label{line:sample_set}
\If{$|\mathcal{S}^{(i)}| < \kappa$}
						\State define estimator $\widetilde{c}^{(i)}_P \gets 0$ \label{line:additive_error}
\Else \label{line:start_bfs}
\State query the $\sigma$ remaining $\ell_0$-samplers
						$\{\mathcal{K}^{(i)}_{j}\}_j$,
						and use the same notations $\{p_j\}_{j=1}^{\sigma}$ as well as
					$\{ (U_j) \}_{j=1}^{\sigma}$,
						to denote the outcome\label{line:notation}
						\State for each $p_j$,
						let $\beta_{p_j} \gets \textsc{BFS}(i, i', p_j, U_j)$
						\State define estimator \label{line:multiplicative_error}
						$ \widetilde{c}^{(i)}_P \gets (\widetilde{n}^{(i)}_S / \sigma)\cdot \sum_{j=1}^{\sigma}{ \indic( \text{\textsc{Is-Relevant}($ i', p_j $) = TRUE} ) \cdot \beta_{p_j} }$
					\EndIf \label{line:end_bfs}
					\EndIf
				\EndFor
				\State query the original MST sketch algorithm of~\cite{FIS08}, and let
				$\widetilde{\MST}(S)$ be the outcome
				\State \textbf{return} $\widetilde{\MST} \gets \Theta\left(\frac{\poly(\epsilon)}{k\log\Delta}\right) \widetilde{\MST}(S) + \widetilde{c}^{(0)}_P - W + \epsilon \cdot \sum_{i=0}^{\log_{1 + \epsilon}{W}}{ (1 + \epsilon)^i \cdot \widetilde{c}^{(i)}_P }$
					\label{line:MSTestimator}
					
				\Comment{the $\widetilde{\MST}(S)$ term is used to make sure $\widetilde{\MST}$ is never an underestimation w.h.p.}
			\EndProcedure
			\Procedure{BFS}{$i, i', p, U$} \State let $U'\gets \{ q\in U: \text{\textsc{Is-Relevant}($i'$, $q$ )} = \text{TRUE} \}$ be the restriction of the neighborhood of $p$ to the query
					\Comment{recall that $U$ contains all non-empty grid points of $\mathcal{G}_{i'}$ at distance $\le \Gamma\cdot (1+\epsilon)^i$ from $p$}
				\State sample integer $Y$ according to distribution $\Pr[Y \geq m] = \frac{1}{m}$
				\State if $Y \geq \Gamma$ or the component in the $(1+\epsilon)^i$-threshold graph on $U'$
				that contains $p$ has more than $Y$ vertices,
				set $\beta \gets 0$, otherwise set $\beta \gets 1$
				\State \textbf{return} $\beta$
			\EndProcedure
			\Procedure{Is-Relevant}{$i', p$}
			\State \textbf{return} TRUE if (i) $\exists x \in A$ s.t.\ portal $x$ belongs to the level-$i'$ quadtree square corresponding to $p$, or (ii) $\exists R_j$ (which is a query cell) s.t.\ $p\in R_j$; otherwise return FALSE
			\EndProcedure
		\end{algorithmic}
	\end{algorithm}
	
	\paragraph{Space and time analysis.}
	As can be seen from \cref{alg:mst_add_sketch}, the
	space is dominated by the $2\sigma$ extended $\ell_0$-samplers of \cref{lemma:neighbor_l0} with $\Gamma = \epsilon^{-3} \log{\Delta}$
	for each $0 \leq i \leq \log_{1 + \epsilon}{\Delta}$
	(both the sparse recovery algorithm $\mathcal{A}^{(i)}$ of Lemma~\ref{lemma:streaming_recover} and the $\ell_0$-norm estimator $\mathcal{N}^{(i)}$
	have smaller space cost).
	Thus, the space bound follows from \cref{lemma:neighbor_l0} and from the value of $\sigma$, defined in~\eqref{eqn:param_def}.

	The update time	is also dominated by maintaining $2\sigma$ extended $\ell_0$-samplers of \cref{lemma:neighbor_l0},
	which can be done in time $\poly(\log k\cdot\Gamma\cdot \epsilon^{-1} \log \Delta)$ for each sampler, including
	the updates to the associated neighborhoods. Thus, the total update time is $k^2 \cdot \poly(\log k\cdot\epsilon^{-1} \log \Delta)$.
	Finally, the query time is bounded by querying all $\ell_0$-samplers and executing the stochastic-stopping BFS from at most $\sigma$ sampled points, each with at most $\Gamma$ steps, which overall takes time of $k^2 \cdot \poly(\epsilon^{-1} \log k \log\Delta)$.	
	
	\subsection{Error Analysis}
	\label{sec:mst_add_error}
	We analyze the error of the query procedure (\cref{alg:mst_add_query}), using the notation defined in \cref{alg:mst_add_sketch,alg:mst_add_query}.	
	First, we show that procedure \textsc{Is-Relevant} captures the points in $P$ exactly, that is,
	for every $i$ and $p \in \mathring{G}^{(i)}_S$, it holds that $p \in \mathring{G}^{(i)}_P$ if and only if
	\textsc{Is-Relevant}($i', p$) returns TRUE.
	As we round to the centers of level-$i'$ quadtree squares, it suffices to shows that these level-$i'$ squares are not larger
	than the cells of the query. Recall that the simple square containing these cells is of level $i_0$ and
	that $i \leq \log_{1 + \epsilon}{W} \leq \log_{1 + \epsilon}{((1+\epsilon)\cdot 2^{i_0})}$ by~\eqref{eqn:diameter_bound}.
	As the cells have sidelength (at least) $\Theta(2^{i_0} \cdot \epsilon)$ (see \cref{sec:dp_review}),
	we have
	\begin{equation}\label{eqn:relevanceTesting}
	2^{i'} \leq O(\epsilon \cdot (1 + \epsilon)^i)\leq O(\epsilon\cdot (1 + \epsilon) \cdot 2^{i_0}) \leq \Theta(2^{i_0} \cdot \epsilon)\,,
	\end{equation}
	using the definition of $i'$ in~\eqref{eqn:iprime} and that the constant hidden in the $O$ notation in~\eqref{eqn:iprime} is sufficiently small, compared to the constant hidden in $\Theta$ in cell sidelength.
	This enables us to filter out the irrelevant samples from $S$ and work
	only on $\mathring{G}^{(i)}_P$, i.e., the points in $P$ rounded to the $2^{i'}$-grid $\mathcal{G}_{i'}$.

	Fix some $i$.
	We bound the error for the estimators $\widetilde{c}^{(i)}_{P}$.
	First, by the guarantee of the $\ell_0$-norm estimator $\mathcal{N}^{(i)}$
		(see line~\ref{line:mathcaln}),
		we know that it uses space $\poly(\log k\cdot \epsilon^{-1}\log\Delta)$
		to achieve with probability at least $1 - \exp(-\log k\cdot \poly(\epsilon^{-1} \log \Delta))$,
		\begin{equation}
			\widetilde{n}^{(i)}_S \in (1 \pm \epsilon) \cdot \mathring{n}^{(i)}_S.
			\label{eqn:tilden_hatn}
		\end{equation}
	We assume this happens, and the probability that it does not happen can be charged to the total failure probability.
	Similarly, we assume the success of all other sketches,
	namely $\mathcal{A}^{(i)}$ and $\mathcal{K}^{(i)}_j$'s, and we require the failure probability to be at most $\exp(-\log k\cdot \poly(\epsilon^{-1}\log \Delta))$.

First, assume that there are at most $\sigma$ non-empty grid points at level $i$, i.e.,
	$\mathring{n}^{(i)}_S \le \sigma$.
	As we assume that $\mathcal{A}^{(i)}$ succeeds, it recovers all the non-empty grid points at level $i$ and then indeed, we can compute $\mathring{c}^{(i)}_P$ exactly
	(recall that $\mathring{c}^{(i)}_P$ already accounts for the error introduced
	by rounding).

Now assume that $\mathring{n}^{(i)}_S > \sigma$.
We have two cases for $\widetilde{c}^{(i)}_P$ in \cref{alg:mst_add_query}
	depending on whether or not we collect enough samples $\mathcal{S}^{(i)}$ that pass the relevance test.
	However, whether or not enough samples are collected is a random event,
	which is not easy to handle if we do the case analysis on it.
	Therefore, we turn our attention to a tightly related quantity $\mathring{n}^{(i)}_P / \mathring{n}^{(i)}_S$, which is the fraction of non-empty grid points of $\mathcal{G}_{i'}$ in the query instance $P$.
	We analyze in \cref{lemma:np_smaller_ns} the error for the estimator when this fraction is small (and most likely, not enough samples are collected),
	and then, in \cref{lemma:np_larger_ns}, the estimation when $\mathring{n}^{(i)}_P / \mathring{n}^{(i)}_S$ is large. \begin{lemma}
		\label{lemma:np_smaller_ns}
Assume that $\mathring{n}^{(i)}_S > \sigma$.
For every $\lambda > 0$,
		if $\mathring{n}^{(i)}_P \leq \lambda \cdot \mathring{n}^{(i)}_S$,
		then $(1 + \epsilon)^i \cdot \mathring{c}^{(i)}_P \le O((\lambda/\epsilon) \cdot \MST(S))$,
		which implies that the estimator $\widetilde{c}^{(i)}_P = 0$ satisfies
		\begin{align*}
			(1 + \epsilon)^i \cdot \mathring{c}^{(i)}_P
			- O(\lambda / \epsilon \cdot \MST(S))
			\leq (1 + \epsilon)^i \cdot \widetilde{c}^{(i)}_P = 0
			\leq (1 + \epsilon)^i \cdot \mathring{c}^{(i)}_P.
		\end{align*}
	\end{lemma}
	\begin{proof}
Using $\mathring{c}^{(i)}_P\le \mathring{n}^{(i)}_P$ and the condition of the lemma, we get
		\begin{align*}
			(1 + \epsilon)^i \cdot \mathring{c}^{(i)}_P
			\leq (1 + \epsilon)^i \cdot \mathring{n}^{(i)}_P
			\leq (1 + \epsilon)^i \cdot \lambda \cdot \mathring{n}^{(i)}_S
			\leq O(\lambda / \epsilon \cdot \MST(S))\,,
		\end{align*}
		where the last inequality follows from Lemma~\ref{mst:lb},
using that $\mathring{n}^{(i)}_S > \sigma\ge 8$.
\end{proof}

	\begin{lemma}
		\label{lemma:np_larger_ns}
Assume that $\mathring{n}^{(i)}_S > \sigma$.
		For every $\lambda \ge 8/\sigma$,
if $\mathring{n}^{(i)}_P \geq \lambda \cdot \mathring{n}^{(i)}_S$,
		then with probability at least $1 - \exp\left(
				- \Omega(\sigma) \cdot
					\left(\frac{\lambda}{\Gamma}
					\right)^2
			\right)$,
		the estimator $\widetilde{c}^{(i)}_P$ in line \ref{line:multiplicative_error} of \cref{alg:mst_add_query} satisfies
		\begin{align*}
			| \widetilde{c}^{(i)}_P - \mathring{c}^{(i)}_P | \leq O(\epsilon)
			\cdot \mathring{c}^{(i)}_P + O\left( \frac{\MST(P)}{\epsilon\cdot\Gamma \cdot (1 + \epsilon)^i} \right).
		\end{align*}
	\end{lemma}
	\begin{proof}
		We use some of the notations from \cref{alg:mst_add_query},
		and suppose $p_j$ and $U_j$ are those from line~\ref{line:notation}.
		For the ease of notation, let
		$\mathcal{E}_j$ be the event that \textsc{Is-Relevant}($i', p_j$) = TRUE.
		Using a similar calculation as in~\cite{FIS08}, we observe that for every $p_j$,
\begin{align}
			\E[ \indic(\mathcal{E}_j) \cdot \beta_{p_j} ]
			&= \Pr[\mathcal{E}_j] \cdot \Pr[\beta_{p_j} = 1 \mid \mathcal{E}_j] \nonumber \\
			&= \frac{\mathring{n}^{(i)}_P}{\mathring{n}^{(i)}_S} \cdot
			\sum_{\text{component H in $\mathring{G}^{(i)}_P$}}{ \Pr[p_j \in H \mid \mathcal{E}_j] \cdot \Pr[Y \geq |H| \land Y < \Gamma]  } \nonumber \\
			&= \frac{\mathring{n}^{(i)}_P}{\mathring{n}^{(i)}_S} \cdot
			\sum_{\text{component H in $\mathring{G}^{(i)}_P$}}{ \frac{|H|}{\mathring{n}^{(i)}_P} \cdot  \Pr[Y \geq |H| \land Y < \Gamma]} \nonumber \\
			&= \frac{1}{\mathring{n}^{(i)}_S} \cdot
			\sum_{\text{component H in $\mathring{G}^{(i)}_P$}}{ |H|  \cdot  \Pr[Y \geq |H| \land Y < \Gamma]},
			\label{eqn:exp_single}
		\end{align}
        recalling that $\beta$'s are binary variables.
		Then, we get an upper bound on $\E[\widetilde{c}^{(i)}_P]$ using~\eqref{eqn:tilden_hatn},
		\begin{align*}
			\E[\widetilde{c}^{(i)}_P]
			&= \frac{\widetilde{n}^{(i)}_S}{\sigma } \sum_{j=1}^{\sigma}{
					\E[ \indic(\mathcal{E}_j) \cdot \beta_{p_j} ]
			} \\
			&= \frac{\widetilde{n}^{(i)}_S}{\sigma \mathring{n}^{(i)}_S} \sum_{j=1}^{\sigma}{
				\sum_{\text{component $H$ in $\mathring{G}^{(i)}_P$}}
				{
					|H| \cdot \Pr[Y \geq |H| \land Y < \Gamma]
				}
			} \\
			&\leq \frac{1 + \epsilon}{\sigma} \cdot \sigma\cdot {
				\sum_{\text{component $H$ in $\mathring{G}^{(i)}_P$}}
				{
					|H| \cdot \Pr[Y \geq |H|]
				}
			} \\
			&= (1 + \epsilon) \cdot \mathring{c}^{(i)}_P\,,
		\end{align*}
		where the last step follows from the distribution of $Y$.
		Similarly, we also get a lower bound
		\begin{align*}
			\E[\widetilde{c}^{(i)}_P]
			&= \frac{\widetilde{n}^{(i)}_S}{\sigma } \sum_{j=1}^{\sigma}{
					\E[ \indic(\mathcal{E}_j) \cdot \beta_{p_j} ]
			} \\
			&= \frac{\widetilde{n}^{(i)}_S}{\sigma \mathring{n}^{(i)}_S} \sum_{j=1}^{\sigma}{
				\sum_{\text{component $H$ in $\mathring{G}^{(i)}_P$}}
				{
					|H| \cdot \Pr[Y \geq |H| \land Y < \Gamma]
				}
			} \\
			&\geq \frac{1 - \epsilon}{\sigma} \cdot \sigma\cdot{
				\sum_{\text{component $H$ in $\mathring{G}^{(i)}_P$}}
				{
					|H| \cdot \left(\frac{1}{|H|} - \frac{1}{\Gamma}\right)
				}
			} \\
			&\geq (1 - \epsilon) \cdot \mathring{c}^{(i)}_P - \frac{1 - \epsilon}{\Gamma} \cdot \mathring{n}^{(i)}_P.
\end{align*}
Next, we apply Hoeffding's inequality. To this end, consider random variables
		$Z_j = \widetilde{n}^{(i)}_S / \sigma \cdot \indic(\mathcal{E}_j) \cdot \beta_{p_j}$ for $j = 1,\dots,\sigma$; note that they are independent
		and we have that $\widetilde{c}^{(i)}_P = \sum_{j=1}^\sigma Z_j$.
		Therefore, by Hoeffding's inequality, it holds that
		\begin{align*}
			\Pr\left[
				\left| \widetilde{c}^{(i)}_P - \E[\widetilde{c}^{(i)}_P] \right|
				\geq \Omega\left( \frac{ \mathring{n}^{(i)}_P } {\Gamma} \right)
			\right]
			&\leq 2 \exp\left(
				- 2\cdot \Omega\left( \frac{ \mathring{n}^{(i)}_P } {\Gamma} \right)^2\cdot \frac{1}{\sigma}
				\cdot \left(\frac{\sigma}{ \widetilde{n}^{(i)}_S} \right)^2
			\right)
			\\
			&= \exp\left(
				- \Omega(\sigma) \cdot
					\left(\frac{\mathring{n}^{(i)}_P}{ \Gamma \widetilde{n}^{(i)}_S}
					\right)^2
			\right)
			\leq \exp\left(
				- \Omega(\sigma) \cdot
					\left(\frac{\lambda}{\Gamma}
					\right)^2
			\right).
		\end{align*}
		where in the second inequality, we use $\mathring{n}^{(i)}_P \geq \lambda \cdot \mathring{n}^{(i)}_S\geq \lambda \cdot (1-\epsilon)\cdot \widetilde{n}^{(i)}_S$ by the assumption of the lemma and by~\eqref{eqn:tilden_hatn}.
		Combining the expectation bound and the above concentration inequality,
		we conclude that with probability at least $1 - \exp\left(
				- \Omega(\sigma) \cdot
					\left(\frac{\lambda}{\Gamma}
					\right)^2
			\right)$,
		\begin{align*}
			| \widetilde{c}^{(i)}_P - \mathring{c}^{(i)}_P |
			\leq O(\epsilon) \cdot \mathring{c}^{(i)}_P
				+ O\left( \frac{ \mathring{n}^{(i)}_P } {\Gamma} \right)
			\leq O(\epsilon) \cdot \mathring{c}^{(i)}_P
				+ O\left( \frac{ \MST(P) } {\epsilon\cdot\Gamma \cdot (1 + \epsilon)^i} \right),
		\end{align*}
		where the last inequality follows from \cref{mst:lb}; 
here we use that $\mathring{n}^{(i)}_P\ge 8$,
		which holds by the assumptions of the lemma
		as $\mathring{n}^{(i)}_P \geq \lambda \cdot \mathring{n}^{(i)}_S > \lambda \cdot \sigma \ge 8$.
\end{proof}
	Next, we do the following case analysis,
still assuming that $\mathring{n}^{(i)}_S > \sigma$.
Let $\lambda_1 := \frac{\kappa}{2\sigma}$, and $\lambda_2 := \frac{2\kappa}{\sigma}$; note that $\lambda_1 \leq \lambda_2$.

	\paragraph{Case I: $\mathring{n}^{(i)}_P < \lambda_1 \cdot \mathring{n}^{(i)}_S$.}
	We claim that with high probability, $|\mathcal{S}^{(i)}| < \kappa$.
	This can be done by using Hoeffding's inequality, as shown in
	\cref{claim:hoeffding_small}.
	\begin{claim}
		\label{claim:hoeffding_small}
		If $\mathring{n}^{(i)}_P < \lambda_1 \cdot \mathring{n}^{(i)}_S$, then $\Pr[|\mathcal{S}^{(i)}| \geq \kappa] \leq \exp(-\log k\cdot \poly(\epsilon^{-1}\log{\Delta}))$.
	\end{claim}
	\begin{proof}
		Let $Z_j$ be the $\{0, 1\}$ random variable, that takes $1$
		if \textsc{Is-Relevant}($i', p_j$) = TRUE.
		Then
		\begin{equation*}
			\Pr[Z_j = 1] = \frac{\mathring{n}^{(i)}_P}{\mathring{n}^{(i)}_S}.
		\end{equation*}
		So $|\mathcal{S}^{(i)}| = \sum_{j=1}^{\sigma}{ Z_j }$,
		and hence,
		\begin{equation*}
			\E[\mathcal{S}^{(i)}] = \frac{\sigma \mathring{n}^{(i)}_P}{\mathring{n}^{(i)}_S} < \sigma \lambda_1
			= \frac{\kappa}{2}.
		\end{equation*}
		By Hoeffding's inequality,
		\begin{align*}
			\Pr[|\mathcal{S}^{(i)}| \geq \kappa]
			&= \Pr[|\mathcal{S}^{(i)}| - \E[ |\mathcal{S}^{(i)}| ]
				\geq \kappa- \E[\mathcal{S}^{(i)}]
			] \\
			&\leq \Pr\left[|\mathcal{S}^{(i)}| - \E[ |\mathcal{S}^{(i)}| ]
				\geq  \frac{\kappa}{2}\right] \\
			&\leq \exp\left(-\frac{\kappa^2}{2\sigma}\right)
			= \exp(-\log k\cdot \poly(\epsilon^{-1}\log\Delta)),
		\end{align*}
		where the last inequality follows from~\eqref{eqn:kappa_sigma_relation}.
\end{proof}
	Then, we assume $|\mathcal{S}^{(i)}| < \kappa$ happens in Case~I,
	and by using \cref{lemma:np_smaller_ns} with $\lambda = \lambda_1$,
	we conclude that the estimator $\widetilde{c}^{(i)}_P = 0$ satisfies
	\begin{align}\label{eqn:tilde_c_caseI}
		(1 + \epsilon)^i \cdot \mathring{c}^{(i)}_P - O\left(\frac{\lambda_1}{\epsilon} \MST(S)\right)
		\leq (1 + \epsilon)^i \cdot \widetilde{c}^{(i)}_P
		\leq (1 + \epsilon)^i \cdot \mathring{c}^{(i)}_P.
	\end{align}
	The case $|\mathcal{S}^{(i)}| \ge \kappa$ only happens with a very small probability,
	and we charge it to the total failure probability of \cref{lemma:streaming_mst_add}.

	\paragraph{Case II: $\mathring{n}^{(i)}_P > \lambda_2 \cdot \mathring{n}^{(i)}_S$.}
	Similarly to Case I, we claim that with probability $1 - \exp(-\log k\cdot \poly(\epsilon^{-1}\log{\Delta}))$,
	$|\mathcal{S}^{(i)}| \geq \kappa$.
	This can be done again by using Hoeffding's inequality,
	and we omit the details since it is very similar to \cref{claim:hoeffding_small}.
	Then, assuming $|\mathcal{S}^{(i)}| \ge \kappa$ and
using \cref{lemma:np_larger_ns} with $\lambda = \lambda_2$,
	we conclude that with probability $1 - \exp\left(
				- \Omega(\sigma) \cdot
					\left(\frac{\lambda_2}{\Gamma}
					\right)^2
			\right)$, the estimator in line \ref{line:multiplicative_error} of \cref{alg:mst_add_query} satisfies
	\begin{align}\label{eqn:tilde_c_caseII}
			| \widetilde{c}^{(i)}_P - \mathring{c}^{(i)}_P | \leq O(\epsilon)
			\cdot \mathring{c}^{(i)}_P + O\left( \frac{\MST(P)}{\epsilon\cdot\Gamma \cdot (1 + \epsilon)^i} \right).
	\end{align}
	\paragraph{Case III:}
	None of the other two cases happens,
	so
	$ \lambda_1 \cdot \mathring{n}^{(i)}_S \leq \mathring{n}^{(i)}_P \leq \lambda_2 \cdot \mathring{n}^{(i)}_S$.
	Then we cannot decide with high probability which type of estimate
	for $\widetilde{c}^{(i)}_P$ the algorithm uses.
	However, since we have both an upper and a lower bound for $\mathring{n}^{(i)}_P / \mathring{n}^{(i)}_S$,
	\cref{lemma:np_smaller_ns,lemma:np_larger_ns} can both be applied
	with a reasonable guarantee.
	In particular, we apply \cref{lemma:np_smaller_ns} with $\lambda = \lambda_2$, which for the estimator $\widetilde{c}^{(i)}_P = 0$
	implies
	\begin{align*}
		(1 + \epsilon)^i \cdot \mathring{c}^{(i)}_P - O\left(\frac{\lambda_2}{\epsilon} \MST(S)\right)
		\leq (1 + \epsilon)^i \cdot \widetilde{c}^{(i)}_P
		\leq (1 + \epsilon)^i \cdot \mathring{c}^{(i)}_P\,.
	\end{align*}
	Next, \cref{lemma:np_larger_ns} with $\lambda = \lambda_1$
	yields with probability at least $1 - \exp\left(
				- \Omega(\sigma) \cdot
					\left(\frac{\lambda_1}{\Gamma}
					\right)^2
			\right)$,
	\begin{align*}
			| \widetilde{c}^{(i)}_P - \mathring{c}^{(i)}_P | \leq O(\epsilon)
			\cdot \mathring{c}^{(i)}_P + O\left( \frac{\MST(P)}{\epsilon\cdot\Gamma \cdot (1 + \epsilon)^i} \right)
	\end{align*}
	for the estimator in line \ref{line:multiplicative_error} of \cref{alg:mst_add_query}.
	Since we do not know which estimator the algorithm actually uses,
	to bound the error, we need to take the worse bound for both directions of these two inequalities, and for the failure probability as well.
	Thus, it holds with probability at least $1 - \exp\left(
				- \Omega(\sigma) \cdot
					\left(\frac{\lambda_1}{\Gamma}
					\right)^2
			\right)$ that
	\begin{align}\label{eqn:tilde_c_caseIII_UB}
		(1 + \epsilon)^i \cdot \widetilde{c}^{(i)}_P
		\leq (1 + O(\epsilon)) \cdot (1 + \epsilon)^i \cdot \mathring{c}^{(i)}_P + O\left(\frac{\MST(P)}{\epsilon\cdot\Gamma}\right),
	\end{align}
	and that
	\begin{align}\label{eqn:tilde_c_caseIII_LB}
		(1 + \epsilon)^i \cdot \widetilde{c}^{(i)}_P
		\geq (1 - O(\epsilon)) \cdot (1 + \epsilon)^i \cdot \mathring{c}^{(i)}_P - O\left( \frac{\MST(P)}{\epsilon\cdot\Gamma} \right) - O\left(\frac{\lambda_2}{\epsilon} \MST(S)\right).
	\end{align}

	\paragraph{Conclusion of the error analysis.}
	Overall, we bound $\widetilde{c}^{(i)}_P$ using the worse bound in both directions (similarly as in Case~III).
Namely, using \cref{eqn:tilde_c_caseI,eqn:tilde_c_caseII,eqn:tilde_c_caseIII_UB,eqn:tilde_c_caseIII_LB} and the union bound,
	except for a failure probability at most
	\begin{align}
			&\quad O(\log_{1 + \epsilon}{W}) \cdot \left(
				\exp(-\log k\cdot \poly(\epsilon^{-1}\log \Delta))
				+ \exp\left(
				- \Omega(\sigma) \cdot
					\left(\frac{\lambda_1}{\Gamma}
					\right)^2
			\right) \right) \nonumber\\
			&\leq \exp(-\log k\cdot \poly(\epsilon^{-1}\log \Delta))
			+ \exp\left(
				-\Omega\left(
					\frac{\kappa^2}{\sigma\Gamma^2}
				\right)
			\right) \nonumber\\
			&\leq \exp(-\log k\cdot \poly(\epsilon^{-1}\log \Delta))\,, \label{eqn:tilde_c_failure_prob}
	\end{align}
	where the last inequality follows from~\eqref{eqn:kappa_sigma_relation},
for any $i = 0,\dots, \log_{1 + \epsilon}{W}$, it holds that
	\begin{align*}
		(1 + \epsilon)^i \cdot \widetilde{c}^{(i)}_P
		\leq (1 + O(\epsilon)) \cdot (1 + \epsilon)^i \cdot \mathring{c}^{(i)}_P + O\left(\frac{\MST(P)}{\epsilon\cdot\Gamma}\right)\,,
	\end{align*}
	and that
	\begin{align*}
		(1 + \epsilon)^i \cdot \widetilde{c}^{(i)}_P
		\geq (1 - O(\epsilon)) \cdot (1 + \epsilon)^i \cdot \mathring{c}^{(i)}_P - O\left( \frac{\MST(P)}{\epsilon\cdot\Gamma} \right) - O\left(\frac{\lambda_2}{\epsilon} \MST(S)\right)\,.
	\end{align*}
(Note that when $\mathring{n}^{(i)}_S \le \sigma$, both of
		these inequalities hold trivially as $\widetilde{c}^{(i)}_P = \mathring{c}^{(i)}_P$.)
Summing over $i$, except for a failure probability
        bounded as in~\eqref{eqn:tilde_c_failure_prob}, we have
	\begin{align*}
		\sum_{i=0}^{\log_{1 + \epsilon}{W}}&{(1 + \epsilon)^i \cdot \widetilde{c}^{(i)}_P}
		\\
		&\geq \sum_{i=0}^{\log_{1 + \epsilon}{W}}{
			(1 - O(\epsilon)) \cdot (1 + \epsilon)^i \cdot \mathring{c}^{(i)}_P
			- O\left( \frac{\MST(P)}{\epsilon\cdot\Gamma} \right) - O\left(\frac{\lambda_2}{\epsilon} \MST(S)\right)
		} \\
		&\geq (1 - O(\epsilon)) \cdot
		\left(\sum_{i=1}^{\log_{1 + \epsilon}{W}}
		{
			(1 + \epsilon)^i \cdot \mathring{c}^{(i)}_P
		} \right)
		- O\left( \frac{\log \Delta}{\epsilon^2 \Gamma}  \right) \cdot \MST(P)
		- O\left( \frac{\lambda_2 \log{\Delta}}{\epsilon^2}\right) \cdot \MST(S)  \\
		&\geq (1 - O(\epsilon)) \cdot
		\left(\sum_{i=1}^{\log_{1 + \epsilon}{W}}
		{
			(1 + \epsilon)^i \cdot \mathring{c}^{(i)}_P
		}\right)
		- O(\epsilon) \cdot \MST(P) - O\left( \frac{\kappa \log{\Delta}}{\epsilon^2 \sigma} \right) \cdot \MST(S) \\
		&\geq (1 - O(\epsilon)) \cdot
		\left( \sum_{i=1}^{\log_{1 + \epsilon}{W}}
		{
			(1 + \epsilon)^i \cdot \mathring{c}^{(i)}_P
		} \right)
		- O(\epsilon) \cdot \MST(P) - O\left( \frac{\poly(\epsilon)}{k \log\Delta} \right) \cdot \MST(S),
	\end{align*}
	where we use $W \leq O(\Delta)$, $\Gamma = \epsilon^{-3}\cdot \log{\Delta}$, and the last inequality is by the definition of $\sigma = k \poly(\epsilon^{-1}\log \Delta) \cdot \kappa$.
	Similarly, for the other direction,
	\begin{align*}
		\sum_{i=0}^{\log_{1 + \epsilon}{W}}{(1 + \epsilon)^i \cdot \widetilde{c}^{(i)}_P}
		&\leq \sum_{i=0}^{\log_{1 + \epsilon}{W}}{
			(1 + O(\epsilon)) \cdot (1 + \epsilon)^i \cdot \mathring{c}^{(i)}_{P}
			+ O\left( \frac{\MST(P)}{\epsilon\cdot\Gamma} \right)
		}  \\
		&\leq (1 + O(\epsilon))
		\left( \sum_{i=1}^{\log_{1 + \epsilon}{W}}{
			(1 + \epsilon)^{i} \cdot \mathring{c}^{(i)}_P
		} \right)
		+ O\left(\frac{\log\Delta}{\epsilon^2 \Gamma}\right) \cdot \MST(P) \\
		&\leq (1 + O(\epsilon))
		\left( \sum_{i=1}^{\log_{1 + \epsilon}{W}}{
			(1 + \epsilon)^{i} \cdot \mathring{c}^{(i)}_P
		} \right)
		+ O(\epsilon) \cdot \MST(P),
	\end{align*}
	where the last inequality is by the definition of $\Gamma = \epsilon^{-3} \log{\Delta}$.
We plug in these two bounds together with \cref{lemma:num_component_relation} into 
our estimate $\widetilde{\MST}$ in line~\ref{line:MSTestimator} of Algorithm~\ref{alg:mst_add_query}.
	Finally, we apply the MST formula in~\eqref{eqn:mst_component} to show the error guarantee of \cref{lemma:streaming_mst_add},
	which concludes the proof of this lemma.\qed

\section{Lower Bound: $\Omega(k)$ Space is Necessary}
\label{sec:lb}

In this section we demonstrate that streaming algorithms for SFP achieving any finite approximation ratio for SFP require $\Omega(k)$ bits of space.

\begin{theorem}
\label{thm:lb}
For every $k > 0$,
every randomized streaming algorithm achieving a finite appro\-xi\-ma\-tion ratio
for SFP with $k$ color classes of size at most $2$
must require $\Omega(k)$ bits of space.
This holds even for insertion-only algorithms and even when points are from the one-dimensional line~$\mathbb{R}$.
\end{theorem}

\begin{proof}
        The proof is a reduction from the INDEX problem on $k$ bits, where Alice holds a binary string $x\in\{0,1\}^k$,
        and Bob has an index $i\in[k]$. The goal of Bob is to compute the bit $x_i$
        in the one-way communication model, where only Alice can send a message to Bob and not vice versa.
        It is well-known that Alice needs to send $\Omega(k)$ bits for Bob to succeed with constant probability \cite{KNR99} (see also~\cite{KushilevitzNisan97,JKS08}).
        Our reduction is from INDEX to SFP on the (discretized) one-dimensional line $[2k]$.
        Consider a randomized streaming algorithm ALG for SFP that approximates the optimal cost
        and in particular can distinguish whether the optimal cost is $0$ or $1$ with constant probability.
        We show that it can be used to solve the INDEX problem, implying that ALG needs to use $\Omega(k)$ bits of space.

        Indeed, Alice applies ALG on the following stream:
        For each bit $x_j$, she adds to the stream a point of color $j$ at location $2j+x_j$.
        So far $\OPT = 0$.
        She now sends the internal state of ALG to Bob.
        Then, Bob continues the execution of ALG (using the same random coins) by adding one more point to the stream:
        Given his index $i\in[k]$,
        he adds a point of color $i$ at location $2i$.
        After that, $\OPT = 0 + x_i$, which is either $0$ or $1$.
        It follows that if ALG achieves a finite approximation with constant probability,
        then Bob can discover $x_i$ and solve INDEX.
\end{proof}

\section{Conclusions and Future Directions}
\label{sec:conclusions}

Our paper makes progress on geometric streaming algorithms and particularly on the applica\-bility of Arora's framework to low-space streaming algorithms for geometric optimization problems. Still, our work leaves a number of very interesting open problems. 

Our approximation ratio $\alpha_2 + \eps$ matches the current approximation ratio for the Steiner tree problem in geometric streams. Hence, any improvement to our approximation ratio would require to first improve the approximation for Steiner tree, even in insertion-only streams.
This naturally leads to the main open problem of obtaining a $(1+\eps)$-approximation for Steiner tree in geometric streams using only $\poly(\eps^{-1}\log\Delta)$ space.

Our naive algorithm for SFP, in Theorem~\ref{thm:easy}, 
achieves the same $(\alpha_2 + \eps)$-approximation with $\poly(k \eps^{-1} \log\Delta)$ space,
but with query time exponential in $k$, because it queries an (approximate) MST-value oracle on all possible subsets of color classes to find the minimum.
We do not know if a smaller number of queries suffices here, but it is known that in a similar setup for coverage problems any oracle-based $O(1)$-approximation requires exponentially many queries to an approximate oracle \cite{DBLP:conf/spaa/BateniEM17}. Thus, it would not be surprising if a similar lower bound holds for our problem.

Our Theorem~\ref{thm:lb} shows that for SFP with color classes of size at most $2$ one cannot achieve any bounded approximation ratio using space that is sublinear in $n\le 2k$. This strongly suggests that SFP with pairs of terminals (i.e., $C_i = \{s_i,t_i\}$) does not admit a constant-factor approximation in the streaming setting, although our lower bound proof does not extend to this case (it relies on having some color classes of size $1$). We leave it as an open problem whether a constant-factor approximation in sublinear (in $n=2k$) space is possible for this version. We notice however that for the case where both points of each terminal pair are inserted/deleted together, it is possible to get an $O(\log n)$-approximation using the metric embedding technique of Indyk~\cite{Indyk04}.

Looking at how our space complexity bounds depend on $k$
(ignoring time complexity) reveals a gap 
between the $O(k^2)$ in \cref{thm:easy} 
and the $\Omega(k)$ in \cref{thm:lb}.
(This question was mentioned to us by David Woodruff.)
In an effort to tighten this gap, 
one may consider instead the sum-of-MST objective,
which asks to minimize the sum of costs of MST trees (i.e., without Steiner points), 
such that points of the same color are in the same tree.
This variant of the question may be more accessible, 
because our naive algorithm in \cref{thm:easy} achieves $(1+\epsilon)$-approximation for sum-of-MST,
and thus does not involve $\alpha_2$ as for SFP.

The focus of this paper is on SFP in the Euclidean plane, but in principle, our entire analysis can be extended to the Euclidean space $\RR^d$, for any fixed $d \ge 2$. However, this would require extending the arguments of~\cite{DBLP:journals/algorithmica/BateniH12, DBLP:journals/talg/BorradaileKM15}, namely, the structural result restated in Theorem~\ref{thm:dp_struct}, and these details were not written explicitly in~\cite{DBLP:journals/algorithmica/BateniH12, DBLP:journals/talg/BorradaileKM15}.

The techniques developed in this paper seem to be general enough to be applicable to other problems/objectives with \emph{connectivity constraints}, where the connectivity is specified by the colors and a solution is feasible if the points of the same color are connected.
One such closely related problem is the aforementioned sum-of-MST objective
(see also \cite{DBLP:journals/ijcga/AnderssonGLN02,ZNI05} for related problems).
We hope that the approach developed in our paper can lead to a $(1+\eps)$-approximation of the geometric version of this problem, using $\poly(k\eps^{-1}\log\Delta)$ space and time (for space only, one can use similar techniques as in Theorem~\ref{thm:easy}).
Moreover, it may be possible to apply our approach to solve the connectivity-constrained variants of other classical problems, especially those where dynamic programming has been employed successfully, like $r$-MST and TSP~\cite{DBLP:journals/jacm/Arora98}. For example, the TSP variant could be to find a collection of cycles of minimum total length
with points of the same color in the same cycle.

At a higher level, the connectivity constraints may be more generally interpreted as grouping constraints. For instance, in the context of clustering, our color constraints may be viewed as \emph{must-link} constraints, where points of the same color have to be placed in the same cluster. Such constrained clustering framework is of significant interest in data analysis
(see, e.g., \cite{DBLP:conf/icml/WagstaffCRS01}).
Our framework, combined with coreset techniques~\cite{FS05} and Arora's quadtree methods (see \cite{DBLP:conf/stoc/AroraRR98}), may be used to design streaming algorithms for such clustering problems.
In fact, optimization problems with connectivity and grouping constraints
are interesting on their own, beyond the streaming or standard offline settings,
and deserve to be studied also in other settings, such as online algorithms, approximation algorithms, fixed-parameter tractability, and heuristics.

\bibliographystyle{alphaurl}
\begin{small}	
\bibliography{ref}
\end{small}

\newpage
\appendix
\section{Technical Lemmas}
\begin{lemma}
\label{lemma:streaming_recover}
There exists an algorithm that, given an integer $n$,
a stream of dynamic updates to a frequency vector $U \in [-n, n]^n$,
and an (integer) threshold $T \geq 1$,
with probability at least $1 - 1/\poly(n)$
it reports YES if $0 < \|U\|_0 \leq T$ and NO if $\|U\|_0 > 2T$ (otherwise, the answer may be arbitrary),
and if the answer is YES, it returns all the non-zero coordinates of $U$,
using space $O(T \cdot \polylog n)$ and the same time per update.
The memory contents of the algorithm is a linear sketch of $U$.
\end{lemma}

\begin{proof}
The algorithm maintains an $\ell_0$-norm estimator for $U$ (see e.g.~\cite{DBLP:conf/pods/KaneNW10}) with relative error $\epsilon=0.5$ using space $\polylog{n}$,
and a compressed sensing structure that
recovers $2T$ non-zero elements of $U$ (see e.g.~\cite{DBLP:conf/sirocco/CormodeM06}) using space $O(T \polylog{n})$,
each succeeding with probability at least $1 - 1/\poly(n)$.
When the stream ends, the algorithm queries the $\ell_0$-norm estimator,
which is accurate enough to distinguish whether $\|U\|_0 \le T$ or $\|U\|_0 > 2T$.
Conditioned on this estimator succeeding,
if it is determined that $\|U\|_0 \leq 2T$,
the algorithm uses the compressed sensing structure to recover the at most $2T$
non-zero coordinates of $U$.
Both the $\ell_0$-norm estimator and the compressed sensing algorithm
are based on linear sketching.
This finishes the proof of \cref{lemma:streaming_recover}.
\end{proof}

We remark that an alternative algorithm is to maintain
$O(T \cdot \polylog n)$ independent instances of
an $\ell_0$-sampler (cf.~\cref{lemma:neighbor_l0})
and applying a coupon collector argument
to both estimate the $\ell_0$ norm and recover the non-zero coordinates. 
The connection goes also the other way: An $\ell_0$-sampler may be designed
using this lemma together with an appropriate subsampling of the domain.

 \section{Composing MST Sketches: $k^k$ Query Time $k^2$ Space}\label{sec:easy_alg}

	As outlined in the introduction, one can solve SFP in a simple way with query time $O(k^k)\cdot \polylog \Delta$.
	In this section, we provide details of this approach and prove the following theorem:

\begin{theorem}
	\label{thm:easy}
		For any integers $k, \Delta \geq 1$ and any $0 < \epsilon < 1/2$,
		one can with high probability $(\alpha_2 + \epsilon)$-approximate
		SFP cost of an input $X \subseteq [\Delta]^2$ presented as a dynamic 
		geometric stream, using space and update time of $O(k^2\cdot \poly(\epsilon^{-1}\cdot \log \Delta)) $
		and with query time $O(k^k) \cdot \poly(\epsilon^{-1}\cdot \log\Delta)$.
\end{theorem}
	
	The proof uses the streaming algorithm for MST from~\cite{FIS08} in a black box manner.
	Namely, we use the following result\footnote{While Theorem~6 in~\cite{FIS08} does not explicitly state that
		the algorithm produces a linear sketch, this follows as the data structures maintained in the algorithm are all linear sketches.
		We describe the MST sketch in Section~\ref{sec:proof_mst_add}.}:
	
	\begin{theorem}[Theorem~6 in~\cite{FIS08}]\label{thm:MSTsketch}
	There is an algorithm that for every $\epsilon, \delta \in (0,1)$, integer $\Delta\ge 1$,
	given a (multi)set $X \subseteq [\Delta]^2$ of points
	presented as a dynamic geometric stream, computes a $(1+\epsilon)$-approximate estimate
	for the cost of the Euclidean minimum spanning tree of $X$
	with probability at least $1-\delta$,
	using space $O\left(\log(1/\delta)\cdot \poly(\epsilon^{-1}\cdot \log \Delta)\right)$ and with both update and query times bounded by the same quantity.
	Furthermore, the algorithm returns a linear sketch from which the estimate can be computed.
	\end{theorem}
	
	\begin{proof}[Proof of \cref{thm:easy}]
	First, we compute the MST sketch $\mathcal{K}_C$ for each color $C$ separately, by using Theorem~\ref{thm:MSTsketch}
	with the same $\epsilon$, with $\delta = 2^{-k} / 3$, and also with the same random bits (so that we are able
	to add up the sketches for different colors). 
	After processing the stream and obtaining sketches $\mathcal{K}_C$, we enumerate all subsets of $k$ colors
	and estimate the MST cost for all colors in the subset.
	Namely, for each subset $S\subseteq [k]$, we first merge (copies of) sketches $\mathcal{K}_C$ for $C\in S$, using that they
	are linear sketches, to get an MST sketch $\mathcal{K}_S$ for all points of colors in $S$.
	Running the estimation procedure from~\cite{FIS08} on sketch $\mathcal{K}_S$,
	we get an estimate of the MST cost for $S$ and store this estimate in memory.
	Then, we enumerate all $O(k^k)$ partitions of $k$ colors and for each partition $I_1, \dots, I_r \subseteq [k]$, we 
	estimate its cost by summing up the estimates for subsets $I_1, \dots, I_r$, computed in the previous step.
	Finally, the algorithm returns the smallest estimate of a partition.
	
	By the union bound, all estimates for subsets of $[k]$ are 	$(1+\epsilon)$-approximate with probability 
	at least $1 - \delta\cdot 2^k \ge \frac{2}{3}$, by the choice of $\delta$.
	Conditioning on this, the cost of any partition is at least $1-\epsilon$ times the optimal cost.
	Consider an optimal solution $F$ of SFP. After removing all Steiner points, we obtain an $\alpha_2$-approximate solution $F'$,
	by the definition of the Steiner ratio $\alpha_2$.
	The partition of colors into components in $F'$ is considered by the algorithm and the estimated 
	cost of this partition is at most $(1+\epsilon)\cdot w(F')\le (1+\epsilon)\cdot \alpha_2\cdot w(F)$ with probability at least $\frac{2}{3}$,
	which implies the correctness of the algorithm.
	The space and time bounds follow from the choice of $\delta = O(2^{-k})$ and from Theorem~\ref{thm:MSTsketch}.
      \end{proof}

\end{document}